%% file: main.tex
\documentclass{article}

\usepackage[preprint]{neurips_2026}

\usepackage{float}
\usepackage{placeins}

\usepackage[utf8]{inputenc}
\usepackage[T1]{fontenc}
\usepackage{hyperref}
\usepackage{url}
\usepackage{booktabs}
\usepackage{amsfonts}
\usepackage{amsmath}
\usepackage{amssymb}
\usepackage{mathtools}
\usepackage{amsthm}
\usepackage{nicefrac}
\usepackage{microtype}
\usepackage{xcolor}
\usepackage{graphicx}
\usepackage{subcaption}
\usepackage{placeins}
\usepackage{algorithm}
\usepackage{algpseudocode}
\usepackage{textcomp}
\usepackage{enumitem}
\usepackage[capitalize,noabbrev]{cleveref}

\theoremstyle{plain}
\newtheorem{theorem}{Theorem}[section]
\newtheorem{proposition}[theorem]{Proposition}

\newtheorem{corollary}[theorem]{Corollary}
\theoremstyle{definition}

\theoremstyle{remark}
\newtheorem{remark}[theorem]{Remark}

\title{MAGIC: Multi-Step Advantage-Gated Causal Influence for Multi-agent Reinforcement Learning}

\author{%
\normalfont
Haohan Yu\textsuperscript{1} \quad
Lu Wang\textsuperscript{2} \quad
Jinmiao Cong\textsuperscript{1} \quad
Shengzhi Wang\textsuperscript{1} \quad
Chanjuan Liu\textsuperscript{1,*}\\[2pt]
\textsuperscript{1}Dalian University of Technology \quad
\textsuperscript{2}Microsoft, Beijing, China\\[2pt]
\texttt{15040488807@mail.dlut.edu.cn} \quad
\texttt{wlu@microsoft.com} \quad
\texttt{cjm111@mail.dlut.edu.cn}\\
\texttt{wangshengzhi@mail.dlut.edu.cn} \quad
\texttt{chanjuanliu@dlut.edu.cn}\\[2pt]
\textsuperscript{*}Corresponding author
}

\begin{document}

\maketitle

\begin{abstract}
A key challenge in multi-agent reinforcement learning (MARL) lies in designing
learning signals that effectively promote coordination among agents. Designing such signals requires estimating how one agent's current action affects
its teammates over future interaction steps. To address this, we introduce Multi-step Advantage-Gated Interventional
Causal MARL (MAGIC), a framework that estimates multi-step action effects
between agents and selectively converts them into intrinsic rewards.
MAGIC uses counterfactual action interventions to compare teammate futures under factual and counterfactual branches, and introduces a gate based on advantage to direct exploration toward beneficial behaviors aligned with the task goal. Experiments on Multi-Agent Particle Environments (MPE) and StarCraft micromanagement benchmarks (SMAC and SMACv2) show that MAGIC consistently outperforms leading prior methods, with average relative final performance improvements of 26.9\% and 10.1\%, respectively.
\end{abstract}

\section{Introduction}

Cooperative multi-agent reinforcement learning (MARL) studies how multiple agents learn policies to optimize a shared team objective~\citep{lowe2017multi,rashid2018qmix,yu2022surprising}. 
In this setting, each agent acts from its local observation, but the value of its action often depends on how it affects other agents and the future team outcome. 
This makes coordination difficult when team rewards are sparse or delayed, because an action that helps a teammate may not receive immediate feedback from the environment~\citep{ma2022elign,liu2023lazy,devidze2022exploration,forbes2024potential}. 
Centralized training with decentralized execution (CTDE) is widely used to alleviate this difficulty~\citep{lowe2017multi,rashid2018qmix,yu2022surprising}. 
It allows the learner to use centralized information during training while keeping each agent's policy decentralized during execution. 
However, the feedback provided by the team reward remains indirect for each individual action. 
The shared return can evaluate the outcome of a joint behavior, but it does not directly indicate whether one agent's current local action will help its teammates coordinate in the future. 
This motivates an additional training signal that can make the cooperative value of local actions more explicit.

One common approach is to use inter-agent influence as an intrinsic reward~\citep{jaques2019social,ma2022elign,li2022pmic,du2024scic}. 
Early influence-based methods encourage coordination by measuring how much one agent affects others through immediate action responses, trajectory dependence, or mutual information~\citep{jaques2019social,li2022pmic,jiang2024mace,liu2024iie}. 
These signals show that influence can be useful for cooperative learning, but they also leave two key issues. 
First, useful influence may be delayed. 
A current action may have little effect on a teammate in the next transition, but may change the teammate's future position, route, or attack timing after several interaction steps. 
Second, strong influence is not necessarily useful. 
An agent can strongly affect its teammates while still reducing the team return. 
Thus, a useful coordination signal should capture delayed influence and should also be aligned with task improvement~\citep{li2022pmic,forbes2024potential,qin2025gradps}.

\begin{figure}[t]
  \centering
  \setlength{\abovecaptionskip}{2pt}
  \setlength{\belowcaptionskip}{-6pt}
  \includegraphics[width=0.245\linewidth]{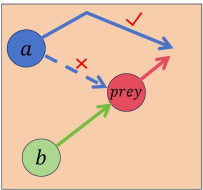}
  \hspace{0.018\linewidth}
  \includegraphics[width=0.245\linewidth]{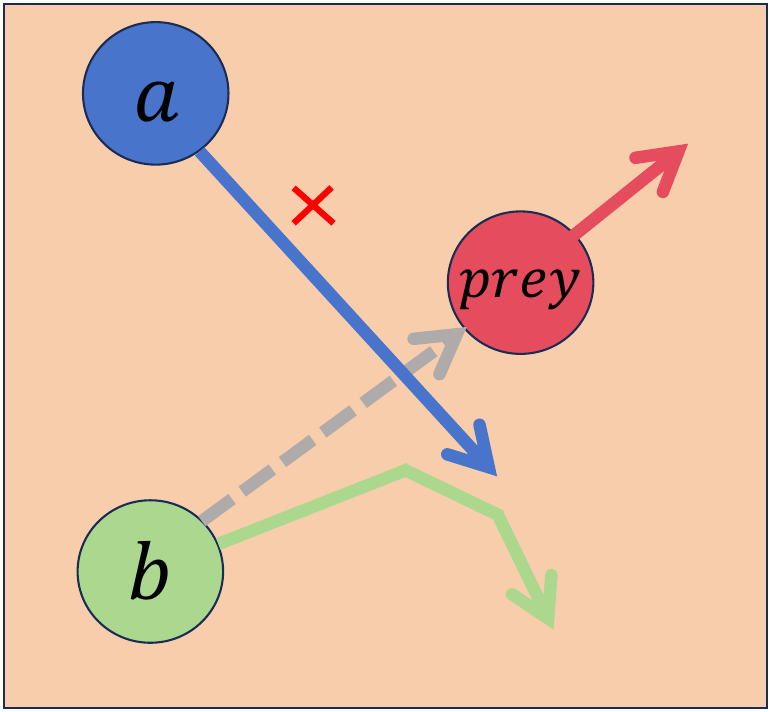}
  \caption{Motivating predator--prey example. Left: a delayed give-way action helps teammate $b$; right: a high-influence action disrupts $b$ and hurts team return.}
  \label{fig:chase-intuition}
\end{figure}

Figure~\ref{fig:chase-intuition} illustrates delayed useful influence and harmful influence in a cooperative predator-prey task.
In the left panel, agent $a$'s give-way action does not immediately bring it closer to the prey, but helps agent $b$ reach a better future capture position.
In the right panel, agent $a$ strongly changes agent $b$'s route, but this interference moves the team away from a better capture outcome.
These examples make the desired training signal more specific.
It should look beyond immediate action effects and encourage such effects only when they support the team objective.

Based on this principle, we propose Multi-step Advantage-Gated Interventional Causal MARL (MAGIC). 
MAGIC augments a standard CTDE learner with a causal and task-aligned intrinsic reward computed only during training. 
To estimate whether one agent's current action affects its teammates' future states, MAGIC constructs factual and counterfactual branches at each sampled time step. 
The factual branch keeps the realized action of the source agent, while the counterfactual branches replace only this source action with valid alternatives. 
All branches share the same current state and the same actions of the other agents. 
MAGIC then compares the teammate futures produced by these branches over multiple rollout steps. 
If replacing the source agent's current action leads to clearly different teammate futures, this indicates that the realized action has a strong counterfactual action effect in the current situation. 
This comparison is different from measuring statistical dependence between actions and future states in collected trajectories. 
It directly asks whether the teammates' future states would change if only the source agent's action were replaced in the same situation.

To compare the teammate futures produced by these branches over multiple rollout steps, MAGIC uses a learned forward model to roll out the factual and counterfactual branches from the same sampled decision context.
In this design, the forward model is not used as an exact long-term simulator.
Its role is to preserve the differences between branches that are relevant to teammate future states within a finite horizon.
When the rollout keeps the relative strength of factual and counterfactual effects separable, MAGIC can estimate action effects reliably even if the absolute predicted states are imperfect.

After the action effect score is obtained, MAGIC aligns it with the task objective through an extrinsic advantage gate.
The action effect score is agent-specific because it measures the effect of each agent's own realized action on its teammates' futures.
The gate uses an extrinsic advantage computed at the team level to estimate whether the realized team transition improves expected task return.
This gate keeps more of the action effect signal when the transition benefits the task and suppresses it when the transition is harmful.
As a result, MAGIC encourages agents to produce effects on teammates only when such effects are supported by task improvement.

Our main contributions are as follows.
\begin{itemize}[leftmargin=*]
    \item We propose a multi-step counterfactual action effect estimator. It measures how replacing one agent's current action changes its teammates' future states over multiple rollout steps, which allows the learning signal to capture delayed coordination effects.
    
    \item We introduce an advantage-gated intrinsic reward that converts action effects into rewards only when they are supported by task improvement. 
The gate retains more intrinsic reward for beneficial team transitions and suppresses high-effect behaviors that reduce the expected task return.
    
    \item We evaluate MAGIC across multiple MARL benchmarks and task families, including
Multi-Agent Particle Environments (MPE), StarCraft Multi-agent Challenge (SMAC), and
SMACv2. MAGIC outperforms representative methods from the two related lines studied
in this paper, including influence-based intrinsic-reward methods and task-aligned
coordination methods. Beyond final performance, we use ablations to identify the role
of each component and diagnostic studies to examine the reliability of the proposed method.
\end{itemize}

\section{Related Work}

Our work studies training signals for improving coordination in cooperative MARL. 
Related efforts include influence-based intrinsic rewards and task-aligned coordination signals.

\noindent\textbf{Influence-based intrinsic rewards.}
Intrinsic rewards are widely used to encourage exploration and coordination under sparse or delayed rewards~\citep{du2019liir,ma2022elign,liu2023lazy,devidze2022exploration,forbes2024potential}. 
In MARL, recent social-influence and coordinated-exploration methods encourage agents to affect each other through mutual information, action response, or interaction-state discovery~\citep{jiang2024mace,liu2024iie}. 
These methods show the value of inter-agent influence, but mainly capture correlation or immediate response. 
SCIC~\citep{du2024scic} estimates single-step interventional causal influence under CTDE and uses it as an intrinsic reward, but it does not capture delayed action effects over multiple rollout steps or distinguish task-beneficial effects from harmful high-effect behaviors. Unlike mutual information based signals, MAGIC directly computes teammate future differences between factual and counterfactual branches, which avoids a separate mutual information critic in the action effect module and makes the score easier to diagnose.

\noindent\textbf{Task-aligned coordination signals.}
Task consistency has long been recognized as important in reward shaping~\citep{devidze2022exploration,wang2023efficient,forbes2024potential}. 
PMIC~\citep{li2022pmic} aligns mutual-information-based coordination with task return by increasing mutual information on high-return trajectories and decreasing it on low-return trajectories. 
However, trajectory-level alignment provides coarse and delayed feedback, so it cannot directly decide whether a specific action effect at the current transition should be encouraged. 
GradPS~\citep{qin2025gradps} improves cooperation by adapting parameter sharing and policy diversity, which helps form more suitable shared representations for coordination. 
However, this adjustment acts at the policy-structure level and does not provide an immediate training signal for evaluating whether an agent's current action benefits its teammates.

\section{Method}

We propose MAGIC, an intrinsic reward module used during training on top of a standard CTDE backbone.
As shown in Figure~\ref{fig:method_overview}, MAGIC compares factual and counterfactual branches that differ only in the source agent's action.
A learned forward model rolls out these branches for a finite horizon, and teammate future differences are aggregated into an action effect score $c_i(t)$.
An extrinsic team advantage gate then filters this score to form the intrinsic reward $r^{\mathrm{int}}_{i,t}$, which is added to the environment reward during training.
The module is removed during execution.
Section~3.1 defines the counterfactual action effect, Section~3.2 describes action effect estimation and aggregation with learned forward model rollouts, and Section~3.3 presents the advantage-gated training objective.

\label{sec:method}

\begin{figure}[t]
  \centering
  \includegraphics[width=\textwidth]{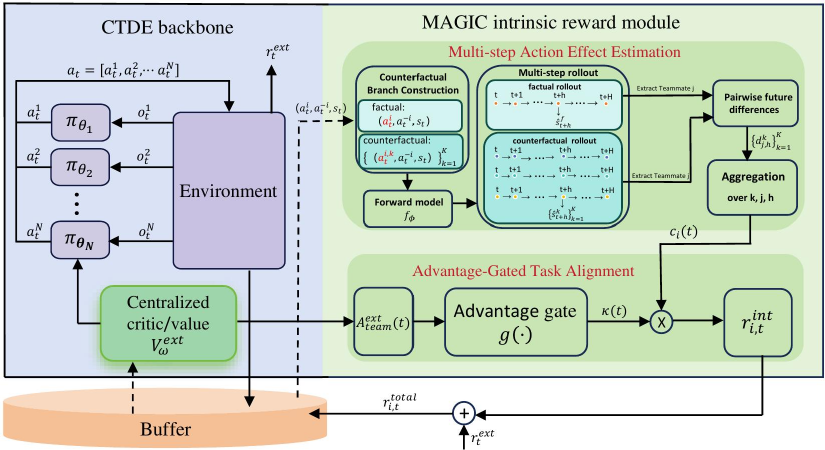}
\caption{Overview of \textsc{MAGIC}. 
\textsc{MAGIC} adds an intrinsic reward module used during training to a CTDE backbone. 
It builds factual and counterfactual branches by replacing only the source agent's action, rolls out predicted centralized states with a learned forward model, extracts teammate future features, and aggregates factual and counterfactual differences into an action effect score. 
An extrinsic team-advantage gate modulates this score before it is added as intrinsic reward. 
The module is removed during execution.}
    \label{fig:method_overview}
\end{figure}

\subsection{Counterfactual Action Effect}
\label{sec:counterfactual_action_effect}

The quantity of interest is the causal effect of a source agent's realized action on its teammates' future states under the same decision context.
This differs from statistical dependence in collected trajectories.
A source action may be correlated with a teammate future simply because it appears in states that already lead to that future.
Such dependence does not tell whether the teammate future would change if only the source action were replaced.

We therefore define the action effect through an intervention that replaces the source action.
For source agent $i$ at time $t$, the decision context consists of the centralized state $s_t$ and the joint action of the other agents $a_t^{-i}$.
The factual branch keeps the realized source action, written as $(a_t^i,a_t^{-i},s_t)$.
The counterfactual branches use the same $a_t^{-i}$ and $s_t$, but replace $a_t^i$ with $K$ valid alternatives, written as $\{(a_t^{i,k},a_t^{-i},s_t)\}_{k=1}^{K}$.
Thus, the only changed component is the source action.

The effect of the realized action is then defined by the difference between teammate futures under the factual and counterfactual branches.
If replacing $a_t^i$ leads to substantially different teammate futures, then the realized action has a strong counterfactual action effect in the current context.
If the teammate futures remain similar after replacement, then the realized action has a weak effect in that context.

\subsection{Multi-Step Action Effect Estimation}
\label{sec:multi_step_action_effect_estimation}

The factual and counterfactual branches defined above specify which action inputs should be compared.
To obtain their future outcomes, MAGIC uses a learned forward model during training.
The forward model takes the joint action and centralized state as input and predicts the next centralized state,
\begin{equation}
    \hat{s}_{t+1} = f_{\phi}(a_t, s_t).
\end{equation}
It is trained from real environment transitions with a one step prediction loss.
During action effect estimation, the model starts from the same sampled decision context and rolls out the factual and counterfactual branches over a finite horizon.
These predicted branch futures are then used for teammate feature comparison and aggregation, rather than being treated as action effect scores directly.
The exact loss and rollout details are given in Appendix~\ref{app:forward-rollout}.

For a source agent $i$, the factual branch starts from $(a_t^i, a_t^{-i}, s_t)$, while the $k$-th counterfactual branch starts from $(a_t^{i,k}, a_t^{-i}, s_t)$. At the first rollout step, the state $s_t$ and the non-source actions $a_t^{-i}$ are shared, and only the source action differs across branches. After this first step, all agents act according to their current policies on the model-predicted states. In this way, the rollout is closed-loop, allowing the effect of the source action to propagate through later policy responses.

The forward model rolls out full centralized states rather than isolated teammate states. We denote the factual rollout states by $\{\hat{s}_{t+h}^{f}\}_{h=1}^{H}$ and the states of the $k$-th counterfactual rollout by $\{\hat{s}_{t+h}^{k}\}_{h=1}^{H}$. For each horizon $h$, teammate features are extracted from these predicted centralized states before computing the effect. Let $z_j(\hat{s})$ denote the normalized future-state feature vector of teammate $j$ extracted from a predicted centralized state $\hat{s}$. The pairwise future difference between the factual branch and the $k$-th counterfactual branch is
\begin{equation}
    d_{j,h}^{(k)}
    =
    \operatorname{dist}
    \left(
    z_j(\hat{s}_{t+h}^{f}),
    z_j(\hat{s}_{t+h}^{k})
    \right),
\end{equation}
where $\operatorname{dist}(\cdot,\cdot)$ is computed in the normalized teammate-feature space. Thus, for each teammate $j$ and horizon $h$, MAGIC obtains a set of pairwise differences $\{d_{j,h}^{(k)}\}_{k=1}^{K}$ by comparing the factual future with each counterfactual future.

The pairwise differences are first averaged over the $K$ counterfactual branches,
\begin{equation}
    d_{j,h}
    =
    \frac{1}{K}
    \sum_{k=1}^{K}
    d_{j,h}^{(k)} .
\end{equation}
The resulting effects are then aggregated over teammates and rollout horizons to obtain the multi-step action effect score for source agent $i$,
\begin{equation}
    \tilde{c}_i(t)
    =
    \sum_{h=1}^{H}
    w_h
    \frac{1}{N-1}
    \sum_{j \neq i}
    d_{j,h},
    \qquad
    w_h \geq 0,\quad
    \sum_{h=1}^{H} w_h = 1 .
\end{equation}
Here, $w_h$ controls the contribution of each rollout horizon. When $H=1$, the score only captures immediate teammate-state changes. When $H>1$, the score can include effects that appear after several interaction steps.

To keep the intrinsic signal on a stable scale, we normalize and clip $\tilde{c}_i(t)$ with running statistics, yielding the scaled action effect score $c_i(t)$ used by the reward module. 
Details of counterfactual action sampling, forward model rollouts, teammate-feature normalization, action effect score scaling, and rollout pseudocode are provided in Appendix~\ref{app:implementation_details}--\ref{app:training-algorithm}.

\subsection{Advantage-Gated Training Objective}
\label{sec:advantage_gated_training_objective}

The scaled action effect score $c_i(t)$ measures how strongly the realized action of agent $i$ changes teammate futures. This effect is not necessarily aligned with the task objective. A large action effect may help teammates coordinate, but it may also disturb their future states in a harmful way. We therefore gate the scaled action effect score with an extrinsic team advantage before using it as intrinsic reward.

The extrinsic team advantage measures whether the realized team transition is beneficial under the original environment reward. Using the centralized value estimate associated with the extrinsic task reward, we compute
\begin{equation}
    A_{\mathrm{team}}^{\mathrm{ext}}(t)
    =
    r_t^{\mathrm{ext}}
    +
    \gamma V_{\omega}^{\mathrm{ext}}(s_{t+1})
    -
    V_{\omega}^{\mathrm{ext}}(s_t).
\end{equation}
The superscript $\mathrm{ext}$ indicates that this advantage is computed with respect to the original task reward rather than the shaped reward. This keeps the gate tied to the task objective.

The gate is defined as a bounded monotone function of the extrinsic team advantage,
\begin{equation}
    \kappa(t)
    =
    g\!\left(A_{\mathrm{team}}^{\mathrm{ext}}(t)\right),
    \qquad
    0 \leq \kappa(t) \leq 1 .
\end{equation}
A larger extrinsic advantage leads to a larger gate value, so action effects observed in beneficial team transitions contribute more to the intrinsic reward. A smaller or negative extrinsic advantage suppresses the contribution of the scaled action effect score.

The intrinsic reward for source agent $i$ is then
\begin{equation}
    r_{i,t}^{\mathrm{int}}
    =
    \lambda_{\mathrm{int}}
    \kappa(t)
    c_i(t),
\end{equation}
where $\lambda_{\mathrm{int}}$ controls the strength of the intrinsic reward. The total reward used for policy learning is
\begin{equation}
    r_{i,t}^{\mathrm{total}}
    =
    r_t^{\mathrm{ext}}
    +
    r_{i,t}^{\mathrm{int}} .
\end{equation}
The gate $\kappa(t)$ is intentionally defined at the team level. MAGIC separates two roles. The score $c_i(t)$ measures whether the realized action of source agent $i$ changes its teammates' future states. The gate $\kappa(t)$ decides whether the realized team transition is aligned with the extrinsic task objective. Thus, the gate is not used as an individual credit-assignment estimator. It acts as a conservative filter for task alignment shared by all agents at the same transition. Agent specificity is still preserved by \(c_i(t)\), since different source agents receive different intrinsic rewards through their own action effect scores.

The CTDE learner is optimized with $r_{i,t}^{\mathrm{total}}$. The forward model is updated with the one-step prediction loss defined in Section~\ref{sec:multi_step_action_effect_estimation}. The forward model, counterfactual rollouts, action effect computation, and intrinsic reward are used only during training. During execution, agents act with their decentralized policies without access to the MAGIC module. Details of the gate implementation, advantage normalization, reward scaling, full training pseudocode, and execution-time behavior are provided in Appendix~\ref{app:training_details}--\ref{app:runtime}.

\subsection{Theoretical Properties and Scope of the Analysis}

We provide the formal properties of MAGIC in Appendix~\ref{app:theory}--\ref{app:gate-theory}.
The analysis focuses on the action effect estimator defined above rather than on a mutual-information surrogate. 
First, we show that the unnormalized multi-step action effect score is nonnegative and becomes zero when factual and counterfactual branches induce identical teammate futures. 
Second, we show that the one-step version can be blind to delayed action effects, while the multi-step score detects effects that appear after several interaction steps. 
Third, we bound the error of the estimated action effect score in terms of forward model rollout errors. 
Finally, we analyze the extrinsic advantage gate as a bounded task-alignment multiplier.

The theoretical statements are separated from the normalization used for numerical stability. 
The exact nonnegativity and zero-effect properties are stated for the unnormalized score \(\tilde c_i(t)\). 
In implementation, we apply a bounded normalization and clipping map before constructing the intrinsic reward. 
Appendix~\ref{app:normalization-theory} shows that this transformation keeps the reward bounded and, when the map is monotone, does not reverse the relative ordering of action effect scores except for possible saturation caused by clipping. 
Thus, the theory describes the underlying action effect estimator, while the practical normalization controls scale without changing the intended comparison between factual and counterfactual branches.

\section{Experiments}
\label{sec:experiments}

We evaluate \textsc{MAGIC} on two benchmark families that cover complementary coordination regimes. 
The first family uses three continuous control tasks from the Multi-agent Particle Environment (MPE)~\citep{lowe2017multi}: \emph{Predator Prey}, \emph{Cooperative Navigation}, and \emph{Cooperative Competitive}. 
These tasks test cooperative pursuit, landmark coverage, and mixed cooperative-competitive interaction in particle world environments. 
The second family uses StarCraft micromanagement benchmarks, including SMAC~\citep{samvelyan2019starcraft} and SMACv2~\citep{ellis2023smacv2}, where agents coordinate under partial observability, discrete actions, and delayed team rewards. 
Across these benchmarks, we first test whether the proposed intrinsic signal improves standard MARL performance, then examine whether the same design remains effective in harder micromanagement tasks, and finally analyze which components drive the gains and when the action effect score estimated with learned forward model rollouts remains reliable.

\paragraph{Baselines and backbones.}
We use matched protocols within each benchmark group for fair comparison. 
On MPE, all methods share the same MADDPG CTDE backbone~\citep{lowe2017multi}, following the common protocol used by prior MPE methods with intrinsic rewards, including SI~\citep{jaques2019social}, PMIC~\citep{li2022pmic}, and SCIC~\citep{du2024scic}. 
MADDPG is the plain backbone without intrinsic rewards.

On SMAC and SMACv2, methods based on policy gradients are evaluated under a unified MAPPO CTDE protocol~\citep{yu2022surprising} implemented in PyMARL2. 
We include MAPPO as the plain policy gradient backbone, QMIX~\citep{rashid2018qmix} as a representative value decomposition baseline, PMIC as a mutual information baseline aligned with task return, SCIC as a single step causal influence baseline using intrinsic rewards, and GradPS~\citep{qin2025gradps} as a parameter sharing method for task aligned coordination. 
Within each benchmark group, methods use matched observation spaces, action masks, training budgets, random seeds, and evaluation procedures. 
Full implementation details are provided in Appendix~\ref{app:exp-details}.

\paragraph{Evaluation protocol.}
All experiments are run with five random seeds under matched training budgets, observation settings, and evaluation procedures within each benchmark group. 
For MPE, we use episodic team return as the main metric, since it is the standard task metric for these continuous-control environments and directly reflects team reward.
We report final return and AUC, where AUC is computed from the team-return learning curve and summarizes learning efficiency over training. 
We also include learning curves to show convergence behavior. 
For SMAC and SMACv2, we use win rate as the main task metric and report final win rate, standard deviation across seeds, and AUC computed from the win-rate learning curve. 
Additional environment configurations, implementation details, hyperparameters, statistical tests, hardware setup, and runtime analysis are provided in Appendix~\ref{app:exp-details}.

\subsection{Benchmark Performance}
\label{subsec:benchmark_performance}

\paragraph{MPE continuous-control tasks.}
Figure~\ref{fig:mpe_curves} compares learning curves on \emph{Predator Prey}, \emph{Cooperative Navigation}, and \emph{Cooperative Competitive}. 
\textsc{MAGIC} improves convergence and final team return across all three MPE tasks. 
With relative gains normalized by the magnitude of the SCIC return, \textsc{MAGIC} improves over SCIC by 26.9\%, 17.5\%, and 36.4\% on \emph{Predator Prey}, \emph{Cooperative Navigation}, and \emph{Cooperative Competitive}, respectively, yielding an average relative final-return gain of 26.9\%.
These results indicate that the proposed signal is useful not only in sparse pursuit but also in dense reward and mixed cooperative-competitive interaction regimes. 
Appendix~\ref{app:mpe-full-results} reports the full MPE scalar metrics, including standard deviations, best returns, AUC, significance tests, and a 10-agent Predator Prey check where \textsc{MAGIC} remains strongest among the compared methods.

\paragraph{SMAC and SMACv2 micromanagement tasks.}
Table~\ref{tab:smac_main} reports final win rate and AUC over five seeds, with standard deviations provided in Appendix~\ref{app:smac-full-results}. 
\textsc{MAGIC} obtains the highest average performance, improving over SCIC by 10.1\% in average final win rate and 16.3\% in average AUC.
The gains are larger on harder maps such as \texttt{corridor}, \texttt{6h\_vs\_8z}, and \texttt{Protoss5v5}, where agents must coordinate under bottlenecks, heterogeneous combat roles, or stronger partial observability. 
Compared with GradPS, \textsc{MAGIC} also achieves higher average win rate and AUC. 
These results suggest that multi-step action effect estimation remains useful when coordination effects are not immediately visible in the next transition.

\begin{figure}[t]
  \centering
  \begin{subfigure}[b]{0.32\linewidth}
    \centering
    \includegraphics[width=\linewidth]{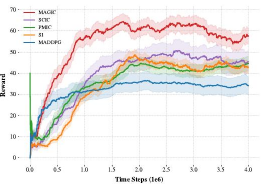}
    \subcaption{Predator Prey}
  \end{subfigure}
  \hfill
  \begin{subfigure}[b]{0.32\linewidth}
    \centering
    \includegraphics[width=\linewidth]{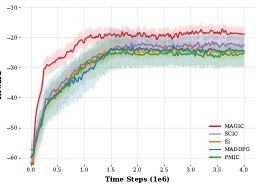}
    \subcaption{Cooperative Navigation}
  \end{subfigure}
  \hfill
  \begin{subfigure}[b]{0.32\linewidth}
    \centering
    \includegraphics[width=\linewidth]{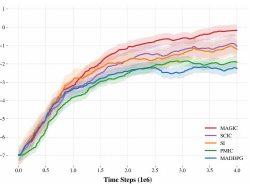}
    \subcaption{Cooperative Competitive}
  \end{subfigure}
  \caption{Learning curves on MPE tasks over five random seeds. Shaded regions denote standard deviation across seeds.}
  \label{fig:mpe_curves}
\end{figure}

\begin{table}[t]
\centering
\caption{Benchmark performance on SMAC and SMACv2 over five random seeds. Each entry reports final Win\% / AUC. AUC is computed from the win-rate learning curve. Standard deviations are provided in Appendix~\ref{app:smac-full-results}.}
\label{tab:smac_main}
\small
\setlength{\tabcolsep}{3pt}
\begin{tabular}{lccccccc}
\toprule
Method & 3s5z & 5m\_vs\_6m & corridor & 6h\_vs\_8z & MMM2 & Protoss5v5 & Avg. \\
\midrule
QMIX  & 92.4 / 78.4 & 78.5 / 52.3 & 48.2 / 25.1 & 58.6 / 32.5 & 74.2 / 42.1 & 36.8 / 15.4 & 64.8 / 41.0 \\
MAPPO & 94.1 / 80.2 & 90.5 / 68.5 & 54.8 / 30.6 & 71.2 / 44.1 & 81.4 / 51.5 & 47.5 / 22.3 & 73.3 / 49.5 \\
PMIC  & 89.3 / 72.1 & 83.7 / 58.4 & 61.2 / 35.2 & 65.8 / 38.6 & 72.1 / 41.8 & 44.3 / 19.5 & 69.4 / 44.3 \\
GradPS & 94.8 / 81.5 & 91.2 / 70.1 & 64.5 / 38.4 & 74.0 / 46.8 & 81.8 / 53.2 & 49.5 / 24.1 & 76.0 / 52.4 \\
SCIC  & 95.2 / 82.3 & 91.8 / 71.4 & 72.6 / 45.8 & 76.5 / 49.2 & 80.3 / 51.8 & 52.4 / 26.5 & 78.1 / 54.5 \\
MAGIC & \textbf{97.8 / 86.5} & \textbf{95.6 / 78.2} & \textbf{83.4 / 55.6} & \textbf{85.2 / 59.4} & \textbf{87.3 / 62.3} & \textbf{66.5 / 38.1} & \textbf{86.0 / 63.4} \\
\bottomrule
\end{tabular}
\end{table}

\subsection{Ablation Study of Multi-Step Action Effects and Advantage Gating}
\label{subsec:component_analysis}

We next examine which parts of \textsc{MAGIC} are responsible for the improvement.
The two components of interest are action effect estimation over multiple rollout steps and advantage gating.
We use two complementary analyses.
The first analysis ablates individual modules on MPE Predator Prey, where the full set of \textsc{MAGIC} components can be cleanly isolated.
The second analysis compares the full method with a one step variant and an ungated variant on representative SMAC maps, testing whether the two components remain useful when coordination is harder and rollout error becomes more relevant.

\begin{table}[t]
\centering
\caption{Component analysis across MPE and SMAC. Panel A reports final and best returns on MPE Predator Prey. Panel B reports final Win\% on representative SMAC maps, with Avg. denoting the average across the three maps. Standard deviations and full MPE scalar metrics are provided in Appendix~\ref{app:component-full-results}.}
\label{tab:component_analysis}
\small
\setlength{\tabcolsep}{4pt}

\begin{tabular}{lccc}
\toprule
\multicolumn{4}{c}{\textbf{Panel A. Module ablation on MPE Predator Prey}} \\
\midrule
Method & Final return & Best return & Change from MAGIC \\
\midrule
MAGIC & \textbf{57.6} & \textbf{62.1} & -- \\
MAGIC w/o Advantage Gating & 47.3 & 48.9 & $-17.9\%$ \\
MAGIC H=1 Action Effect & 41.8 & 43.4 & $-27.4\%$ \\
\bottomrule
\end{tabular}

\vspace{0.6em}

\begin{tabular}{lcccc}
\toprule
\multicolumn{5}{c}{\textbf{Panel B. Core component attribution on SMAC}} \\
\midrule
Method & \texttt{corridor} & \texttt{5m\_vs\_6m} & \texttt{MMM2} & Avg. \\
\midrule
MAGIC H=3 & \textbf{83.4} & \textbf{95.6} & \textbf{87.3} & \textbf{88.8} \\
MAGIC H=3 w/o Gate & 77.5 & 91.8 & 83.8 & 84.4 \\
MAGIC H=1+Gate & 73.1 & 92.4 & 82.1 & 82.5 \\
MAPPO & 54.8 & 90.5 & 81.4 & 75.6 \\
\bottomrule
\end{tabular}

\end{table}

Panel A of Table~\ref{tab:component_analysis} shows that both design components contribute to performance on MPE. 
Replacing the multi-step action effect estimator with a one step version causes the largest drop, which indicates that one-step effects are insufficient for capturing delayed cooperative consequences. 
Removing the advantage gate also reduces performance, but the ungated multi-step variant still remains stronger than the one-step variant. 
This suggests that the multi-step action effect signal itself captures useful delayed coordination information, while the advantage gate further improves this signal by emphasizing effects observed in task-beneficial transitions.

Panel B of Table~\ref{tab:component_analysis} tests whether the same two design choices remain useful on harder SMAC maps. 
Using a one-step action effect signal with the gate improves over MAPPO on all three maps, but remains below the full method. 
Using the multi-step action effect signal without the gate also improves over MAPPO on all three maps, showing that delayed action effect estimation is useful by itself. 
Adding the advantage gate further improves performance on all three maps, indicating that the gate helps convert the delayed action effect signal into a training signal that is better aligned with the task. 
Together, these results suggest that the two components are complementary. 
Multi-step action effect estimation provides the delayed interaction signal, while advantage gating makes this signal more task-aligned and more effective across environments.

\subsection{Reliability and Horizon Sensitivity of Multi-Step Action Effect Estimation}
\label{subsec:score_reliability}

Since MAGIC estimates multi-step action effects through learned forward model rollouts, we evaluate the reliability of this estimation from three complementary views in Table~\ref{tab:reliability}. 
Panel A fixes the main horizon at $H=3$ and tests whether separability of branch effects tracks downstream performance as the forward model update ratio changes. 
Panel B studies the sensitivity to the horizon parameter $H$, which guides the choice of rollout depth. 
Panel C explicitly corrupts the forward model at $H=3$ to stress-test the estimator.

We diagnose the reliability of the learned rollouts from two sides: whether the forward model predicts future states accurately, and whether its predicted branch differences preserve the true strength of action effects. 
\emph{In-MSE} measures prediction error on normal policy rollouts, where all agents follow their learned policies. 
\emph{Int-MSE} measures prediction error on intervention rollouts, where the selected agent's action is replaced while the other agents follow their learned policies. 
These two errors indicate how accurate the learned dynamics are under factual and counterfactual rollout conditions. 
\emph{Sep. AUC} measures whether the model-predicted teammate-future differences reflect the true action effect strength. 
A value near 0.5 indicates random separation between large-effect and small-effect branches. 
A higher value means that the predicted branch differences better identify actions that truly cause larger changes in teammate futures.

\begin{table*}[t]
\centering
\caption{Reliability and horizon sensitivity of multi-step action effect estimation. Panel A fixes $H=3$ and varies the forward model update ratio on \texttt{5m\_vs\_6m}. Panel B reports horizon sensitivity averaged over three maps. Panel C reports explicit forward model corruption on \texttt{5m\_vs\_6m} with $H=3$.}
\label{tab:reliability}
\footnotesize
\renewcommand{\arraystretch}{1.05}

\begin{minipage}[t]{0.36\textwidth}
\centering
\textbf{Panel A. Fixed-horizon reliability}\\[-0.2em]
\resizebox{\linewidth}{!}{%
\begin{tabular}{ccccc}
\toprule
Update ratio & In-MSE & Int-MSE & Sep. AUC & Win\%\\
\midrule
$0.25\times$ & 0.135 & 0.228 & 0.57 & 86.8 \\
$0.50\times$ & 0.072 & 0.108 & 0.74 & 91.2 \\
$0.75\times$ & 0.038 & 0.051 & 0.87 & 94.3 \\
$1.00\times$ & 0.025 & 0.032 & 0.91 & \textbf{95.6} \\
\bottomrule
\end{tabular}%
}
\end{minipage}
\hfill
\begin{minipage}[t]{0.33\textwidth}
\centering
\textbf{Panel B. Horizon sensitivity}\\[-0.2em]
\resizebox{\linewidth}{!}{%
\begin{tabular}{ccccc}
\toprule
$H$ & In-MSE & Int-MSE & Sep. AUC & Win\%\\
\midrule
1  & 0.010 & 0.013 & 0.94 & 82.5 \\
2  & 0.018 & 0.024 & 0.92 & 86.1 \\
3  & 0.031 & 0.039 & 0.90 & \textbf{88.8} \\
5  & 0.065 & 0.088 & 0.82 & 86.3 \\
8  & 0.151 & 0.218 & 0.58 & 74.7 \\
10 & 0.228 & 0.345 & 0.49 & 68.2 \\
\bottomrule
\end{tabular}%
}
\end{minipage}
\hfill
\begin{minipage}[t]{0.25\textwidth}
\centering
\textbf{Panel C. Forward model corruption}\\[-0.2em]
\resizebox{\linewidth}{!}{%
\begin{tabular}{cccc}
\toprule
Noise & Int-MSE & Sep. AUC & Win\%\\
\midrule
0.0 & 0.032 & 0.91 & \textbf{95.6} \\
0.1 & 0.055 & 0.88 & 94.2 \\
0.5 & 0.120 & 0.79 & 91.5 \\
1.0 & 0.280 & 0.55 & 82.4 \\
\bottomrule
\end{tabular}%
}
\end{minipage}

\end{table*}

Panel A of Table~\ref{tab:reliability} fixes the main horizon at $H=3$ and varies only the forward model update ratio. 
As the update ratio increases from $0.25\times$ to $1.00\times$, both in-distribution and intervention rollout errors decrease, Sep. AUC increases from 0.57 to 0.91, and the win rate improves from 86.8 to 95.6. 
With the horizon fixed, this shows that separability of branch effects tracks downstream performance under the main setting.

Panel B of Table~\ref{tab:reliability} is a horizon sensitivity study.
Increasing the horizon from $H=1$ to $H=3$ improves performance while keeping rollout error in a reliable range. 
This indicates that one-step action effects miss useful delayed effects. 
Increasing the horizon further eventually reverses this trend. 
At $H=8$ and $H=10$, intervention error becomes much larger, Sep. AUC approaches random separation, and win rate drops. 
Thus the useful horizon is not simply the largest possible rollout depth. 
It is the range in which delayed interaction effects become visible before model error destroys separability of branch effects. The same rise-then-decline pattern also appears on MPE Predator Prey, with full results reported in Appendix~\ref{app:mpe-horizon}.

Panel C of Table~\ref{tab:reliability} tests robustness by explicitly corrupting the forward model at the main horizon $H=3$. 
Moderate corruption increases intervention MSE, but performance remains strong as long as Sep. AUC stays high. 
When corruption becomes large, Sep. AUC drops toward 0.5 and task performance falls accordingly. 
This supports the view that separability of branch effects is a useful reliability diagnostic for MAGIC, while raw MSE is an upstream diagnostic.

Additional appendix diagnostics examine this reliability criterion under counterfactual sampling, stochasticity, and artificial delay.
Appendix~\ref{app:k-ablation} varies the number of counterfactual branches on continuous MPE Predator Prey.
Appendix~\ref{app:smac-action-slip} and Appendix~\ref{app:stochasticity} test stochasticity in SMAC action execution and MPE dynamics.
Appendix~\ref{app:delay-results} studies artificial reward delays and shows that larger horizons help only while separability remains preserved.
Together, these results show that \textsc{MAGIC} is useful when learned rollouts preserve effect differences between branches, and that longer rollouts or more branch samples help only within this reliable regime.

\section{Conclusion}

We propose \textsc{MAGIC}, a multi-step advantage-gated action effect framework for cooperative MARL.
\textsc{MAGIC} estimates whether one agent's current action changes its teammates' future states over multiple rollout steps, and uses an extrinsic advantage gate to convert this signal into task-aligned intrinsic reward.
This separates action effect estimation from task-value alignment and reduces the risk of rewarding harmful high-effect behaviors.
Experiments across MPE, SMAC, and SMACv2 show that \textsc{MAGIC} achieves the best average performance among strong influence-based and task-aligned coordination baselines.
Ablations and diagnostics show that the gains depend on both multi-step estimation and advantage gating, and that the method remains reliable when learned rollouts preserve separability of branch effects.

\bibliographystyle{plainnat}
\bibliography{references}

\section*{Impact Statement}

This work develops a training-time method for improving coordination in cooperative MARL. 
The experiments are limited to simulated benchmark environments and the method is not deployed in real-world systems. 
As with general MARL methods, real-world use would require domain-specific safety evaluation, especially in settings involving autonomous coordination.

\appendix

\section{Additional Method Details}
\label{app:method-details}

This appendix provides the implementation details and theoretical properties referenced by Section~3. We keep the terminology consistent with the main text: MAGIC computes an agent-specific action effect score $c_i(t)$ from factual and counterfactual teammate futures, and then modulates this score with a team-level extrinsic advantage gate $\kappa(t)$. The forward model, counterfactual rollouts, action effect computation, and advantage gate are used only during training.

\subsection{Counterfactual Branch Construction}
\label{app:implementation_details}

For each sampled transition $(s_t,a_t,r_t^{\mathrm{ext}},s_{t+1})$, MAGIC constructs branch comparisons separately for each possible source agent $i$. The realized joint action is written as $a_t=(a_t^i,a_t^{-i})$, where $a_t^i$ is the action of the source agent and $a_t^{-i}$ denotes the actions of all other agents. The factual branch is
\begin{equation}
    b_t^{f,i}=(a_t^i,a_t^{-i},s_t).
\end{equation}
The $k$-th counterfactual branch is
\begin{equation}
    b_t^{k,i}=(a_t^{i,k},a_t^{-i},s_t), \qquad k=1,\ldots,K,
\end{equation}
where $a_t^{i,k}$ is a valid alternative action for agent $i$. This construction changes only the source action. The centralized state $s_t$ and the non-source actions $a_t^{-i}$ are kept fixed across the factual and counterfactual branches. Therefore the branch difference measures the effect of replacing the source agent's current action under the same decision context.

Counterfactual actions are sampled from the valid action set of the source agent. In discrete-action tasks, invalid actions are excluded using the environment action mask. If the executed action is drawn as a counterfactual candidate, we resample it when another valid action is available, so that each counterfactual branch represents an actual replacement. In continuous-control tasks, counterfactual actions are sampled within the bounded action range of the environment and clipped to the same action limits used by the policy. The number of counterfactual branches $K$ is fixed for a run. Unless otherwise stated, the main experiments use $K=64$. The counterfactual branches are not used to search for an optimal replacement action. They provide a Monte Carlo estimate of how much teammate futures change under valid replacements of the source action. Thus, increasing $K$ improves the coverage of alternative actions, but the estimator does not require dense optimization over the action space.

\subsection{Forward Model Training and Closed-Loop Rollouts}
\label{app:forward-rollout}

The forward model $f_{\phi}$ predicts the next centralized state from the current centralized state and joint action:
\begin{equation}
    \hat{s}_{t+1}=f_{\phi}(a_t,s_t).
\end{equation}
It is trained with real environment transitions and the one-step squared prediction loss
\begin{equation}
    \mathcal{L}_{\mathrm{fm}}(\phi)
    =
    \left\|f_{\phi}(a_t,s_t)-s_{t+1}\right\|_2^2.
\end{equation}
The forward model is not trained to predict the action effect score. It only provides predicted future centralized states for the factual and counterfactual branches.

Rollouts are closed-loop after the first action-replacement step. 
At the first model step, the factual branch uses $(a_t^i,a_t^{-i})$ and the $k$-th counterfactual branch uses $(a_t^{i,k},a_t^{-i})$. 
After this step, each branch evolves with the current decentralized policies applied to the model-predicted state. 
At each closed-loop rollout step, the predicted centralized state is converted into each agent's local observation using the same observation construction as the environment or the benchmark wrapper. 
For discrete-action tasks, the valid action mask from the corresponding predicted or current rollout state is used when available; otherwise invalid actions are excluded according to the benchmark action constraints. 
In SMAC and SMACv2, the predicted centralized state is decoded into unit-level features following the PyMARL2 state representation. We recompute local observations and action availability from these predicted unit features using the same visibility, attack-range, alive-state, and movement-boundary rules as the environment wrapper. Predicted continuous fields are clipped to valid ranges before observation and mask construction. If a unit is predicted as dead, only the no-op action is marked available.
When a policy is stochastic, rollout actions are sampled from the current policy during training. 
For diagnostic evaluation, we use the same deterministic evaluation convention as the corresponding backbone. 
This closed-loop procedure allows the first action replacement to influence later states through subsequent policy responses, rather than only measuring a one-step state perturbation.

For a source agent $i$, the factual rollout produces $\{\hat{s}^{f}_{t+h}\}_{h=1}^{H}$ and the $k$-th counterfactual rollout produces $\{\hat{s}^{k}_{t+h}\}_{h=1}^{H}$. All rollout states are full centralized states. Teammate-specific features are extracted after the rollout, rather than being predicted by separate teammate-specific models.

\subsection{Teammate Feature Extraction and Normalization}
\label{app:feature-normalization}

Let $z_j(\hat{s})$ denote the feature vector extracted for teammate $j$ from a predicted centralized state $\hat{s}$. The extractor uses the centralized state representation available during CTDE training and selects the components associated with teammate $j$. In MPE tasks, these components include task-relevant physical state variables such as position and velocity. In SMAC and SMACv2 tasks, they include the centralized unit features associated with the teammate unit, such as position, health-related state, alive status, and other unit-level state variables available to the centralized critic. The extractor does not use hidden information at execution time, because the entire MAGIC module is training-time only.

Before computing distances, each feature dimension is normalized using running statistics collected from training batches. For a feature vector $z$, we use
\begin{equation}
    \bar{z} = \frac{z-\mu_z}{\sigma_z+\epsilon},
\end{equation}
where $\mu_z$ and $\sigma_z$ are running mean and standard deviation estimates and $\epsilon$ is a small constant for numerical stability. The feature statistics are updated once per training update using batch exponential moving averages with momentum $0.99$, pooled over sampled transitions and agents in the training batch. The same normalization statistics are used for factual and counterfactual branches. The pairwise branch difference is then computed in this normalized feature space. In the main experiments we use the Euclidean distance,
\begin{equation}
    d_{j,h}^{(k)}=
    \left\|
    \bar{z}_j(\hat{s}^{f}_{t+h})-
    \bar{z}_j(\hat{s}^{k}_{t+h})
    \right\|_2 .
\end{equation}
This distance measures how much the predicted future state of teammate $j$ changes when only the source action is replaced.

\subsection{Action Effect Score Aggregation and Scaling}
\label{app:score-scaling}

For each source agent $i$, MAGIC first averages the pairwise differences over the $K$ counterfactual branches:
\begin{equation}
    d_{j,h}=\frac{1}{K}\sum_{k=1}^{K}d_{j,h}^{(k)}.
\end{equation}
It then averages over teammates and aggregates over rollout horizons:
\begin{equation}
    \tilde{c}_i(t)=
    \sum_{h=1}^{H}
    w_h
    \frac{1}{N-1}
    \sum_{j\neq i} d_{j,h},
    \qquad
    w_h\geq 0,
    \quad
    \sum_{h=1}^{H}w_h=1.
    \label{eq:action-effect-aggregation}
\end{equation}
In the main experiments we use uniform horizon weights unless otherwise stated. The quantity $\tilde{c}_i(t)$ is the unnormalized multi-step action effect magnitude. To keep the reward scale stable, we divide it by a running scale estimate and clip the result:

\begin{equation}
c_i(t)=\mathrm{clip}\left(
\frac{\tilde c_i(t)}{\sigma_c+\epsilon},0,c_{\max}
\right).
\label{eq:score-scaling}
\end{equation}
where $\sigma_c$ is a running standard deviation estimate of $\tilde{c}_i(t)$ over training batches. It is updated with the same batch exponential moving average rule used for teammate-feature normalization and is pooled over source agents and sampled transitions. This scaling preserves the non-negativity of the action effect score while preventing a small number of large branch differences from dominating the intrinsic reward.

\subsection{Advantage Gate, Advantage Normalization, and Reward Scaling}
\label{app:training_details}

The gate is based on an extrinsic team advantage. We maintain or reuse a centralized value estimate $V_{\omega}^{\mathrm{ext}}$ trained only with the original environment reward. In the MADDPG-based MPE experiments, $V_{\omega}^{\mathrm{ext}}$ is implemented as a separate centralized state-value network and is trained with one-step TD targets $r_t^{\mathrm{ext}}+\gamma V_{\bar{\omega}}^{\mathrm{ext}}(s_{t+1})$. In the MAPPO-based SMAC and SMACv2 experiments, we maintain an extrinsic-only value head with the same value architecture and update it from extrinsic-return targets. In both cases, this value estimate is never trained with the intrinsic reward or the total shaped reward. The one-step extrinsic TD advantage is
\begin{equation}
    A_{\mathrm{team}}^{\mathrm{ext}}(t)
    =
    r_t^{\mathrm{ext}}
    +
    \gamma V_{\omega}^{\mathrm{ext}}(s_{t+1})
    -
    V_{\omega}^{\mathrm{ext}}(s_t).
\end{equation}
This advantage does not include the intrinsic reward or the total shaped reward. Thus the gate is tied to the original task objective and does not use the shaped reward to decide its own strength.

In implementation, we normalize the scalar advantage before applying the gate:
\begin{equation}
    \bar{A}_{\mathrm{team}}^{\mathrm{ext}}(t)
    =
    \frac{A_{\mathrm{team}}^{\mathrm{ext}}(t)-\mu_A}{\sigma_A+\epsilon},
\end{equation}
where $\mu_A$ and $\sigma_A$ are running statistics over training batches. They are updated once per training update using batch exponential moving averages with momentum $0.99$, pooled over transitions in the training batch. The gate is then
\begin{equation}
    \kappa(t)=g\!\left(\bar{A}_{\mathrm{team}}^{\mathrm{ext}}(t)\right),
\end{equation}
where $g$ is a bounded monotone function. We use a sigmoid gate in the experiments. The normalization is part of the implementation of $g$ and does not change the definition in the main text.

The gate $\kappa(t)$ is shared by all agents at the same time step because it evaluates the task value of the realized team transition. The action effect score $c_i(t)$ remains agent-specific because it is computed from the branch comparisons of source agent $i$. The intrinsic reward is
\begin{equation}
    r_{i,t}^{\mathrm{int}}
    =
    \lambda_{\mathrm{int}}\kappa(t)c_i(t),
\end{equation}
where $\lambda_{\mathrm{int}}$ controls the strength of the intrinsic reward. The total reward used for the policy and critic update is
\begin{equation}
    r_{i,t}^{\mathrm{total}}
    =
    r_t^{\mathrm{ext}}+r_{i,t}^{\mathrm{int}}.
\end{equation}
The shaped reward is kept agent-specific. Each agent's policy and value target use its own $r_{i,t}^{\mathrm{total}}$, while the extrinsic gate is shared across agents. We do not average $r_{i,t}^{\mathrm{int}}$ across agents before constructing policy or value targets.
After $c_i(t)$ is clipped by Eq.~\eqref{eq:score-scaling}, we do not apply an additional intrinsic-reward clipping step in the main experiments.
The intrinsic reward scale is controlled by $c_{\max}$ and $\lambda_{\mathrm{int}}$.

\subsection{Full Training Procedure}
\label{app:training-algorithm}

Algorithm~\ref{alg:magic-training} summarizes the training procedure. The algorithm is written for a generic CTDE backbone. In off-policy experiments, sampled transitions come from a replay buffer. In on-policy MAPPO experiments, they come from the current rollout batch. In both cases, the intrinsic reward is computed from training samples without additional environment interaction.

\begin{algorithm}[H]
\caption{Training procedure of MAGIC}
\label{alg:magic-training}
\begin{algorithmic}[1]
\State Initialize CTDE policies, centralized critic or value functions, extrinsic value estimate $V_{\omega}^{\mathrm{ext}}$, forward model $f_{\phi}$, and training buffer or rollout storage.
\For{each training iteration}
    \State Collect environment transitions $(s_t,a_t,r_t^{\mathrm{ext}},s_{t+1})$ using the current decentralized policies.
    \State Update the forward model $f_{\phi}$ on real one-step transitions with $\mathcal{L}_{\mathrm{fm}}(\phi)$.
    \For{each sampled transition and each source agent $i$}
        \State Construct the factual branch $(a_t^i,a_t^{-i},s_t)$.
        \State Sample $K$ valid counterfactual source actions $\{a_t^{i,k}\}_{k=1}^{K}$ and construct $(a_t^{i,k},a_t^{-i},s_t)$.
        \State Roll out the factual and counterfactual branches with $f_{\phi}$ for $H$ steps using the closed-loop procedure.
        \State Extract normalized teammate features from the predicted centralized states.
        \State Compute $d_{j,h}^{(k)}$, aggregate over $k$, teammates, and horizons, and obtain $c_i(t)$ after scaling and clipping.
        \State Compute $A_{\mathrm{team}}^{\mathrm{ext}}(t)$ using the extrinsic value estimate and form $\kappa(t)$.
        \State Compute $r_{i,t}^{\mathrm{int}}=\lambda_{\mathrm{int}}\kappa(t)c_i(t)$ and $r_{i,t}^{\mathrm{total}}=r_t^{\mathrm{ext}}+r_{i,t}^{\mathrm{int}}$.
    \EndFor
    \State Update the CTDE backbone using $r_{i,t}^{\mathrm{total}}$ with the same actor-critic or policy-gradient update rule as the underlying backbone.
\EndFor
\end{algorithmic}
\end{algorithm}

\subsection{Execution-Time Behavior and Computational Cost}
\label{app:runtime}

MAGIC adds computation only during training. During execution, each agent uses the same decentralized policy as the underlying CTDE backbone. The forward model, counterfactual branches, action effect computation, extrinsic advantage gate, and intrinsic reward are not used at execution time.

The main additional training cost comes from model rollouts. For a batch with $B$ sampled transitions, $N$ agents, $K$ counterfactual branches per source agent, and horizon $H$, the rollout cost scales as $O(BNKH)$ forward model steps, up to constants from feature extraction and distance computation. The cost is controlled by using a moderate horizon and a fixed number of counterfactual branches. For very large agent populations, \textsc{MAGIC} can be applied with sparse source-agent or teammate subsets. Instead of computing the action effect score for every source agent in every sampled transition, one can sample a subset of source agents per batch. Similarly, the teammate aggregation in Eq.~\eqref{eq:action-effect-aggregation} can be restricted to a local neighborhood or communication graph, replacing the average over all $j\neq i$ with an average over a subset $\mathcal{N}(i)$. This changes the computation from all-pair action effect estimation to local action effect estimation while keeping the same factual and counterfactual branch construction. We do not use this approximation in the main experiments because the evaluated benchmarks have moderate agent counts, but it is a natural scaling option for swarm settings. In our measurements, the training overhead is about \(10\%\) relative to SCIC under matched settings, while execution-time cost is unchanged.

\subsection{Theoretical Properties of the Action Effect Module}
\label{app:theory}

This subsection gives the formal properties referenced in Section~\ref{sec:method}. The analysis is stated for the action effect estimator used in the main text. It is stated directly for the factual and counterfactual branch-difference estimator used by MAGIC. We first analyze the unnormalized score $\tilde c_i(t)$, then state what remains true after the bounded normalization and clipping map used to obtain $c_i(t)$, and finally analyze the extrinsic advantage gate.

\paragraph{Notation.}
For a fixed source agent $i$ and time $t$, let $\hat{s}_{t+h}^{f}$ and $\hat{s}_{t+h}^{k}$ denote the model-predicted factual and $k$-th counterfactual centralized states at horizon $h$. The empirical pairwise branch difference is
\begin{equation}
    \hat d_{j,h}^{(k)}
    =
    \operatorname{dist}\left(
    z_j(\hat{s}_{t+h}^{f}),
    z_j(\hat{s}_{t+h}^{k})
    \right),
\end{equation}
where $z_j(\cdot)$ extracts normalized teammate features and $\operatorname{dist}$ is the Euclidean distance in the normalized feature space in our implementation. The unnormalized empirical score is
\begin{equation}
    \hat{\tilde c}_i(t)
    =
    \sum_{h=1}^{H} w_h
    \frac{1}{N-1}
    \sum_{j\neq i}
    \frac{1}{K}
    \sum_{k=1}^{K}
    \hat d_{j,h}^{(k)},
    \qquad
    w_h\ge 0,
    \quad
    \sum_{h=1}^{H} w_h=1.
\end{equation}
When discussing approximation error, we use the same notation without hats, $s_{t+h}^{f}$, $s_{t+h}^{k}$, $d_{j,h}^{(k)}$, and $\tilde c_i^{\star}(t)$, for the corresponding quantities under the true environment dynamics and the same branch construction.

\begin{proposition}[Basic properties of the unnormalized score]
\label{prop:basic-score-properties}
Assume $w_h\ge 0$ and $\sum_{h=1}^{H}w_h=1$, and assume $\operatorname{dist}(x,y)\ge 0$ with equality if and only if $x=y$. Then the unnormalized empirical action effect score satisfies:
\begin{enumerate}[leftmargin=*]
    \item $\hat{\tilde c}_i(t)\ge 0$.
    \item $\hat{\tilde c}_i(t)=0$ if and only if every positive-weight factual and counterfactual teammate feature difference is zero, i.e., for all $h$ with $w_h>0$, all $j\neq i$, and all $k$, we have $z_j(\hat{s}_{t+h}^{f})=z_j(\hat{s}_{t+h}^{k})$.
    \item If there exists a horizon $h$ with $w_h>0$, a teammate $j\neq i$, and a counterfactual branch $k$ such that $z_j(\hat{s}_{t+h}^{f})\neq z_j(\hat{s}_{t+h}^{k})$, then $\hat{\tilde c}_i(t)>0$.
\end{enumerate}
\end{proposition}

\begin{proof}
Each pairwise difference $\hat d_{j,h}^{(k)}$ is nonnegative by the definition of the distance. The score $\hat{\tilde c}_i(t)$ is a nonnegative weighted sum of these nonnegative terms, which proves the first claim. The sum can be zero only if every term with positive coefficient is zero. Since $w_h>0$, $1/(N-1)>0$, and $1/K>0$, this is equivalent to $\hat d_{j,h}^{(k)}=0$ for every positive-weight horizon, teammate, and counterfactual branch. By the identity-of-indiscernibles property of the distance, this holds if and only if the corresponding teammate features are identical. The third claim is the contrapositive of the zero-score characterization.
\end{proof}

\begin{corollary}[Reduction to one-step action effect estimation]
\label{cor:one-step-reduction}
If $H=1$ and $w_1=1$, then MAGIC reduces to a one-step action effect estimator:
\begin{equation}
    \hat{\tilde c}_i(t)
    =
    \frac{1}{N-1}
    \sum_{j\neq i}
    \frac{1}{K}
    \sum_{k=1}^{K}
    \operatorname{dist}\left(
    z_j(\hat{s}_{t+1}^{f}),
    z_j(\hat{s}_{t+1}^{k})
    \right).
\end{equation}
It only measures the effect of replacing the source action on teammate features at the next predicted step.
\end{corollary}

\begin{proof}
Substituting $H=1$ and $w_1=1$ into the definition of $\hat{\tilde c}_i(t)$ removes all horizons except $h=1$, yielding the stated expression.
\end{proof}

\begin{proposition}[One-step blindness to delayed action effects]
\label{prop:delayed-effects}
Consider a fixed decision context and source agent $i$. Suppose there exists an integer $d>1$ such that, for all teammates $j\neq i$ and all counterfactual branches $k$,
\begin{equation}
    z_j(s_{t+1}^{f})=z_j(s_{t+1}^{k}),
\end{equation}
but for some teammate $j^{\star}\neq i$ and some counterfactual branch $k^{\star}$,
\begin{equation}
    z_{j^{\star}}(s_{t+d}^{f})\neq z_{j^{\star}}(s_{t+d}^{k^{\star}}).
\end{equation}
Then the one-step action effect score is zero. In contrast, any multi-step score with $H\ge d$ and $w_d>0$ is strictly positive under the true dynamics.
\end{proposition}

\begin{proof}
For $H=1$, all teammate feature differences at $t+1$ are zero by assumption, so Corollary~\ref{cor:one-step-reduction} and Proposition~\ref{prop:basic-score-properties} imply that the one-step score is zero. If $H\ge d$ and $w_d>0$, then the difference at horizon $d$ for teammate $j^{\star}$ and branch $k^{\star}$ is nonzero. Proposition~\ref{prop:basic-score-properties} then implies that the multi-step score is strictly positive.
\end{proof}

\paragraph{Forward model approximation error.}
The preceding properties describe the branch differences produced by a given rollout process. We now relate the learned-model score to the score that would be obtained under the true dynamics with the same factual and counterfactual branch construction. This statement explains why rollout errors are relevant to the reliability diagnostics in Table~\ref{tab:reliability}.
For the error-bound analysis, we consider deterministic environment dynamics or branch comparisons coupled with the same exogenous randomness. Equivalently, the bounds can be read conditionally on the sampled exogenous noise.

\begin{proposition}[Forward model error bound for action effect scores]
\label{prop:forward-error-bound}
Assume each teammate feature extractor $z_j$ is $L_z$-Lipschitz with respect to the centralized state norm, i.e.,
\begin{equation}
    \|z_j(s)-z_j(s')\|_2 \le L_z\|s-s'\|_2
    \quad \text{for all } j.
\end{equation}
Let the model rollout errors at horizon $h$ be
\begin{equation}
    \epsilon_h^{f}=\|\hat{s}_{t+h}^{f}-s_{t+h}^{f}\|_2,
    \qquad
    \epsilon_h^{k}=\|\hat{s}_{t+h}^{k}-s_{t+h}^{k}\|_2.
\end{equation}
Then each pairwise branch-difference error is bounded by
\begin{equation}
    \left|\hat d_{j,h}^{(k)}-d_{j,h}^{(k)}\right|
    \le
    L_z\left(\epsilon_h^{f}+\epsilon_h^{k}\right),
\end{equation}
and the aggregated score error satisfies
\begin{equation}
\label{eq:score-error-bound}
    \left|\hat{\tilde c}_i(t)-\tilde c_i^{\star}(t)\right|
    \le
    L_z
    \sum_{h=1}^{H} w_h
    \frac{1}{K}
    \sum_{k=1}^{K}
    \left(\epsilon_h^{f}+\epsilon_h^{k}\right).
\end{equation}
\end{proposition}

\begin{proof}
For Euclidean distance, the reverse triangle inequality gives
\begin{align}
    \left|\hat d_{j,h}^{(k)}-d_{j,h}^{(k)}\right|
    &\le
    \left\|
    z_j(\hat{s}_{t+h}^{f})-z_j(s_{t+h}^{f})
    \right\|_2
    +
    \left\|
    z_j(\hat{s}_{t+h}^{k})-z_j(s_{t+h}^{k})
    \right\|_2 .
\end{align}
The Lipschitz property of $z_j$ then gives
\begin{equation}
    \left|\hat d_{j,h}^{(k)}-d_{j,h}^{(k)}\right|
    \le L_z\epsilon_h^{f}+L_z\epsilon_h^{k}.
\end{equation}
Substituting this bound into the weighted average defining the score gives
\begin{align}
    \left|\hat{\tilde c}_i(t)-\tilde c_i^{\star}(t)\right|
    &\le
    \sum_{h=1}^{H}w_h
    \frac{1}{N-1}\sum_{j\neq i}
    \frac{1}{K}\sum_{k=1}^{K}
    \left|\hat d_{j,h}^{(k)}-d_{j,h}^{(k)}\right| \\
    &\le
    L_z
    \sum_{h=1}^{H} w_h
    \frac{1}{K}
    \sum_{k=1}^{K}
    \left(\epsilon_h^{f}+\epsilon_h^{k}\right),
\end{align}
because the bound does not depend on $j$. This proves the claim.
\end{proof}

\subsection{Normalization, Clipping, and Reward-Level Properties}
\label{app:normalization-theory}

The structural properties above are stated for the unnormalized score $\tilde c_i(t)$ or its learned-model estimate $\hat{\tilde c}_i(t)$. In implementation, MAGIC uses a bounded transformation to obtain the reward-level score $c_i(t)$. In the main implementation this transformation is
\begin{equation}
\label{eq:psi-definition}
    c_i(t)=\psi(\hat{\tilde c}_i(t)),
    \qquad
    \psi(x)=\operatorname{clip}\left(\frac{x}{\sigma_c+\epsilon},0,c_{\max}\right),
\end{equation}
where $\sigma_c$ is a running scale estimate, $\epsilon>0$, and $c_{\max}$ is the clipping threshold. This transformation is used for numerical stability. It is not part of the definition of the underlying factual and counterfactual action effect.

\begin{proposition}[Properties preserved by normalization and clipping]
\label{prop:normalization-clipping}
Let $\psi$ be defined as in Equation~\eqref{eq:psi-definition}. Then:
\begin{enumerate}[leftmargin=*]
    \item $0\le c_i(t)\le c_{\max}$ for all $i,t$.
    \item $\psi$ is monotone nondecreasing. Therefore, if $x_1\le x_2$, then $\psi(x_1)\le\psi(x_2)$. Clipping may make two large scores equal after saturation, but it does not reverse their order.
    \item $\psi$ is globally Lipschitz with constant $L_{\psi}=1/(\sigma_c+\epsilon)$. Hence
    \begin{equation}
        |\psi(x)-\psi(y)|\le L_{\psi}|x-y|.
    \end{equation}
\end{enumerate}
\end{proposition}

\begin{proof}
The first claim follows directly from clipping to $[0,c_{\max}]$. The map $x\mapsto x/(\sigma_c+\epsilon)$ is monotone increasing because $\sigma_c+\epsilon>0$, and clipping to an interval is monotone nondecreasing, so their composition is monotone nondecreasing. The clipping map is 1-Lipschitz, and the scaling map has Lipschitz constant $1/(\sigma_c+\epsilon)$; the composition is therefore $1/(\sigma_c+\epsilon)$-Lipschitz.
\end{proof}

\begin{corollary}[Score-error propagation after scaling]
\label{cor:scaled-score-error}
Under the assumptions of Proposition~\ref{prop:forward-error-bound}, the normalized and clipped score error satisfies
\begin{equation}
    \left|\psi(\hat{\tilde c}_i(t))-\psi(\tilde c_i^{\star}(t))\right|
    \le
    L_{\psi}
    L_z
    \sum_{h=1}^{H} w_h
    \frac{1}{K}
    \sum_{k=1}^{K}
    \left(\epsilon_h^{f}+\epsilon_h^{k}\right).
\end{equation}
\end{corollary}

\begin{proof}
Apply the Lipschitz bound in Proposition~\ref{prop:normalization-clipping} to the raw score error bound in Proposition~\ref{prop:forward-error-bound}.
\end{proof}

\begin{remark}[Scope of the zero-effect statement]
\label{rem:zero-effect-scope}
The exact zero-effect characterization in Proposition~\ref{prop:basic-score-properties} is a structural statement about the unnormalized branch-difference score. The implemented score $c_i(t)=\psi(\hat{\tilde c}_i(t))$ remains nonnegative and bounded. Because $\psi(0)=0$, identical factual and counterfactual teammate futures still produce zero intrinsic action effect score. However, due to clipping, very large but different raw scores can map to the same saturated value. We therefore use the structural properties to characterize the underlying estimator, and use Proposition~\ref{prop:normalization-clipping} and Corollary~\ref{cor:scaled-score-error} to characterize the normalized reward-level score.
\end{remark}

\subsection{Properties of the Extrinsic Advantage Gate}
\label{app:gate-theory}

The gate uses the extrinsic team advantage defined in the main text. Let
\begin{equation}
    u(t)=\frac{A_{\mathrm{team}}^{\mathrm{ext}}(t)-\mu_A}{\sigma_A+\epsilon}
\end{equation}
be the normalized gate input used in implementation, and let
\begin{equation}
    \kappa(t)=g(u(t)),
    \qquad
    0\le g(u)\le 1,
\end{equation}
where $g$ is monotone nondecreasing. For fixed running statistics, $u(t)$ is monotone in the raw extrinsic team advantage, so larger extrinsic team advantage gives a larger gate. The same $\kappa(t)$ is shared by all agents at time $t$, while the action effect score $c_i(t)$ remains source-agent specific.

\begin{proposition}[Bounded task-alignment multiplier]
\label{prop:gate-bounds}
Assume $\lambda_{\mathrm{int}}\ge0$, $0\le \kappa(t)\le 1$, and $0\le c_i(t)\le c_{\max}$. Then the intrinsic reward
\begin{equation}
    r_{i,t}^{\mathrm{int}}=\lambda_{\mathrm{int}}\kappa(t)c_i(t)
\end{equation}
satisfies
\begin{equation}
    0\le r_{i,t}^{\mathrm{int}}\le \lambda_{\mathrm{int}}c_{\max}.
\end{equation}
For a fixed $c_i(t)$, $r_{i,t}^{\mathrm{int}}$ is monotone nondecreasing in the gate input $u(t)$, and therefore monotone nondecreasing in $A_{\mathrm{team}}^{\mathrm{ext}}(t)$ for fixed running statistics.
\end{proposition}

\begin{proof}
The bound follows immediately from multiplying the inequalities $0\le \kappa(t)\le 1$ and $0\le c_i(t)\le c_{\max}$ by $\lambda_{\mathrm{int}}\ge0$. If $c_i(t)$ is fixed, then $r_{i,t}^{\mathrm{int}}$ is an affine nonnegative multiple of $\kappa(t)$. Since $\kappa(t)=g(u(t))$ and $g$ is monotone nondecreasing, the intrinsic reward is monotone nondecreasing in $u(t)$. Because $u(t)$ is a positive affine transformation of $A_{\mathrm{team}}^{\mathrm{ext}}(t)$ for fixed running statistics, the same monotonicity holds with respect to the extrinsic team advantage.
\end{proof}

\begin{proposition}[Controlled contribution from low-gate transitions]
Let $\mathcal{L}$ be any set of transitions whose normalized gate input satisfies $u(t) \le \rho$.
Define $\epsilon_\rho = \sup_{u \le \rho} g(u)$. Then for any transition in $\mathcal{L}$,
\[
r^{\mathrm{int}}_{i,t} \le \lambda_{\mathrm{int}}\epsilon_\rho c_{\max}.
\]
In particular, transitions assigned a small gate can contribute only a controlled amount of intrinsic
reward even when their action effect score is large.
\end{proposition}

\begin{proof}
For any transition in $\mathcal{L}$, we have $u(t)\le \rho$. By the definition of
$\epsilon_\rho$, $\kappa(t)=g(u(t))\le \epsilon_\rho$. Combining this with
$c_i(t)\le c_{\max}$ gives the bound.
\end{proof}

\begin{proposition}[Second-moment control by gating]
\label{prop:second-moment}
Let $R_i^{\mathrm{ungated}}(t)=\lambda_{\mathrm{int}}c_i(t)$ and $R_i^{\mathrm{gated}}(t)=\lambda_{\mathrm{int}}\kappa(t)c_i(t)$. If $0\le \kappa(t)\le1$, then
\begin{equation}
    \mathbb{E}\left[(R_i^{\mathrm{gated}}(t))^2\right]
    \le
    \mathbb{E}\left[(R_i^{\mathrm{ungated}}(t))^2\right].
\end{equation}
\end{proposition}

\begin{proof}
For every transition, $|R_i^{\mathrm{gated}}(t)|=\kappa(t)|R_i^{\mathrm{ungated}}(t)|\le |R_i^{\mathrm{ungated}}(t)|$. Squaring both sides and taking expectations proves the claim.
\end{proof}

\begin{proposition}[Bounded perturbation of the extrinsic objective]
\label{prop:bounded-perturbation}
Consider the discounted shaped objective obtained by adding the intrinsic reward to the extrinsic team reward. Suppose the extrinsic objective $J_{\mathrm{ext}}$ has a unique maximizer $\pi^{\star}$ with margin $\Delta>0$, meaning
\begin{equation}
    J_{\mathrm{ext}}(\pi^{\star})\ge J_{\mathrm{ext}}(\pi)+\Delta
    \quad \text{for all } \pi\neq \pi^{\star}.
\end{equation}
If $0\le c_i(t)\le c_{\max}$ and $0\le\kappa(t)\le1$, then the total discounted intrinsic return over $N$ agents is bounded by
\begin{equation}
    0\le J_{\mathrm{int}}(\pi)
    \le
    \frac{N\lambda_{\mathrm{int}}c_{\max}}{1-\gamma}.
\end{equation}
Therefore, if
\begin{equation}
    \frac{N\lambda_{\mathrm{int}}c_{\max}}{1-\gamma}<\Delta,
\end{equation}
then $\pi^{\star}$ remains the unique maximizer of $J_{\mathrm{ext}}+J_{\mathrm{int}}$.
\end{proposition}

\begin{proof}
The per-agent bound in Proposition~\ref{prop:gate-bounds} gives $0\le r_{i,t}^{\mathrm{int}}\le \lambda_{\mathrm{int}}c_{\max}$. Summing over $N$ agents and over discounted time gives
\begin{equation}
    0\le J_{\mathrm{int}}(\pi)
    \le
    \sum_{t=0}^{\infty}\gamma^t N\lambda_{\mathrm{int}}c_{\max}
    =
    \frac{N\lambda_{\mathrm{int}}c_{\max}}{1-\gamma}.
\end{equation}
For any $\pi\neq\pi^{\star}$,
\begin{align}
    (J_{\mathrm{ext}}+J_{\mathrm{int}})(\pi^{\star})
    -
    (J_{\mathrm{ext}}+J_{\mathrm{int}})(\pi)
    &\ge
    \Delta
    -
    J_{\mathrm{int}}(\pi) \\
    &\ge
    \Delta
    -
    \frac{N\lambda_{\mathrm{int}}c_{\max}}{1-\gamma}.
\end{align}
The right-hand side is positive under the stated condition, so no policy $\pi\neq\pi^{\star}$ can overtake $\pi^{\star}$ under the shaped objective.
\end{proof}

\paragraph{Remark on objective shift.}
The intrinsic reward used by \textsc{MAGIC} is a training-time shaping signal, so the shaped training objective is not identical to the original extrinsic objective. Proposition~A.11 should be read as a bounded-perturbation result rather than a general convergence guarantee for function-approximation MARL. It shows that the intrinsic term cannot arbitrarily dominate the extrinsic objective when $\lambda_{\mathrm{int}}$ and $c_{\max}$ are fixed, and that the original optimum is preserved whenever the extrinsic margin is larger than the maximum discounted intrinsic perturbation. In the main experiments, this theoretical control is paired with the extrinsic advantage gate, which further suppresses action effect rewards on low-advantage transitions.

\section{Additional Experimental Details and Results}
\label{app:exp-details}

This appendix provides the experimental details and additional results referenced by Section~\ref{sec:experiments}. We first describe environments, baselines, architectures, training protocols, evaluation metrics, and compute. We then provide the full tables supporting the main experimental claims.

\subsection{Environment Details and Reward Protocols}
\label{app:env-details}

The MPE experiments use three PettingZoo MPE tasks: Predator Prey, Cooperative Navigation, and Cooperative Competitive. All MPE tasks use five learning agents and an episode length of $25$ steps unless otherwise stated. Predator Prey evaluates cooperative pursuit, Cooperative Navigation evaluates landmark coverage and collision avoidance, and Cooperative Competitive evaluates mixed cooperative-competitive interaction. We use the default environment dynamics and reward functions of the benchmark implementation. The team return is computed from the environment rewards and is used as the extrinsic reward for all methods. For the mixed cooperative-competitive MPE task, we follow the cooperative-team evaluation protocol and compute the reported team return over the controlled cooperative agents; the same scalar team reward is used as $r_t^{\mathrm{ext}}$ for all compared methods.

The SMAC and SMACv2 experiments use six micromanagement maps: \texttt{3s5z}, \texttt{5m\_vs\_6m}, \texttt{corridor}, \texttt{6h\_vs\_8z}, \texttt{MMM2}, and \texttt{Protoss5v5}. These tasks evaluate coordination under partial observability, discrete actions, action masking, and delayed team rewards. We use the standard win rate as the main metric. For all benchmark groups, MAGIC and the baselines are evaluated under the same environment reward. The intrinsic reward is used only as an additional training signal.

\paragraph{Existing assets and licenses.}
We use existing benchmark environments and software frameworks only for academic research comparisons. 
The MPE experiments use the PettingZoo implementation of MPE~\citep{terry2021pettingzoo}. 
PettingZoo is publicly released by the Farama Foundation, with Farama-owned code released under the MIT license and included third-party components under their corresponding MIT or Apache-2.0 licenses. 
The SMAC and SMACv2 experiments use the public StarCraft micromanagement benchmark interfaces~\citep{samvelyan2019starcraft,ellis2023smacv2}. 
SMACv2 is released under the MIT license, and the PyMARL2 codebase used for the MAPPO-based protocol is released under the Apache-2.0 license. 
StarCraft II is used only as a third-party simulator through the benchmark interface, and we do not redistribute any StarCraft II game assets. 
All baseline methods and benchmark assets are cited in the paper and are used under their released research terms.

\begin{table}[H]
\centering
\caption{MPE task configurations used in our experiments.}
\label{tab:app-mpe-config}
\small
\begin{tabular}{lccc}
\toprule
Task & Agents & Team structure & Episode length \\
\midrule
Predator Prey & 5 & cooperative predators & 25 \\
Cooperative Navigation & 5 & fully cooperative & 25 \\
Cooperative Competitive & 5 & mixed cooperative-competitive & 25 \\
\bottomrule
\end{tabular}
\end{table}

\subsection{Baselines and Fairness Protocol}
\label{app:baseline-details}

Within each benchmark group, all compared methods use matched training budgets, random seeds, observation settings, and evaluation schedules. On MPE, all methods are implemented on the same MADDPG-based CTDE backbone. MADDPG is the plain backbone without intrinsic rewards. SI, PMIC, SCIC, and MAGIC use the same environment interaction budget and the same replay-buffer training schedule. No intrinsic-reward method receives extra environment interaction.

On SMAC and SMACv2, MAGIC is evaluated with a MAPPO-based CTDE protocol implemented in PyMARL2. MAPPO is included as the plain policy-gradient backbone. QMIX is included as a representative value-decomposition baseline. PMIC, SCIC, and GradPS are evaluated under matched observation spaces, action masks, training budgets, and evaluation procedures whenever applicable. The purpose of this protocol is to compare the coordination signal rather than changes in environment access or evaluation budget. For baseline-specific hyperparameters, we use the official or recommended settings when available, while keeping the shared backbone, training budget, observation space, action masks, and evaluation protocol matched within each benchmark group.

\subsection{Network Architectures and Hyperparameters}
\label{app:hyperparameters}

For MPE, actors use two hidden layers of size $(128,128)$ with ReLU activations and a final tanh squashing layer for bounded continuous actions. The centralized critic and the extrinsic value network use two hidden layers of size $(256,256)$. The extrinsic value network is trained only from environment rewards and supplies $V_{\omega}^{\mathrm{ext}}$ for the gate. The forward model is a two-layer MLP with hidden sizes $(256,256)$ and ReLU activations. It maps the current centralized state and joint action to the next centralized state. MAGIC does not use an additional mutual-information estimator in the action effect module.

For SMAC and SMACv2, we use the default MAPPO configuration from PyMARL2. 
The recurrent actor uses a one-layer GRU with hidden dimension 64, and we keep the default centralized value architecture and PPO update settings. 
The SMAC/SMACv2 forward model uses the same two-layer MLP architecture with hidden sizes $(256,256)$ and ReLU activations. Its input is the centralized state concatenated with one-hot joint actions, and its output is the next centralized state.
MAGIC introduces no modifications to the backbone policy network, action masking, or PPO objective. The only value-side addition is the extrinsic-only value head used for the gate.
MAGIC adds only the training-time action effect module: a forward model, counterfactual action-replacement branches, teammate-future feature differences, action effect aggregation, and the extrinsic team-advantage gate. 
The policy used at execution is unchanged from the MAPPO backbone.

\begin{table}[H]
\centering
\caption{Main hyperparameters for the MADDPG-based MPE experiments.}
\label{tab:app-mpe-hyperparams}
\small
\begin{tabular}{lc}
\toprule
Hyperparameter & Value \\
\midrule
Discount factor $\gamma$ & 0.95 \\
Episode length & 25 \\
Total environment steps per task & $4\times 10^6$ \\
Optimizer & Adam \\
Learning rate for actor, critic, and auxiliary networks & $1\times 10^{-3}$ \\
Replay buffer size & $10^6$ transitions \\
Mini-batch size & 1024 \\
Target-network update & Polyak, $\tau=0.01$ \\
Gradient clipping & global norm 5.0 \\
Main rollout horizon $H$ & 3 \\
Number of counterfactual branches $K$ & 64 \\
Intrinsic weight $\lambda_{\mathrm{int}}$ & $0.05$ \\
Number of random seeds & 5 \\
Evaluation episodes per checkpoint & 10 \\
Evaluation policy & deterministic \\
\bottomrule
\end{tabular}
\end{table}

For SMAC and SMACv2, we use StarCraft II version 4.10.0 and a unified MAPPO protocol. The shared settings are Adam with learning rate $5\times 10^{-4}$, 8 parallel workers, 3200 collected transitions per update cycle, 5 PPO epochs, PPO clipping ratio 0.2, and GAE parameter $\lambda=0.95$. Each method is trained for 10 million environment steps without map-specific tuning. Evaluation is performed every 10,000 environment steps using 32 deterministic test episodes.

Table~\ref{tab:magic_hyperparams} lists the MAGIC-specific hyperparameters used in the main experiments. We keep the intrinsic-reward hyperparameters fixed across tasks and vary them only in the diagnostic studies explicitly reported below. The $1.00\times$ forward model update ratio in Table~\ref{tab:reliability} corresponds to this main setting. The ratios $0.25\times$, $0.50\times$, $0.75\times$, and $1.00\times$ correspond to 2, 5, 8, and 10 forward model epochs per RL iteration, respectively.

\begin{table}[H]
\centering
\caption{Core MAGIC-specific hyperparameters used in the main experiments. The same intrinsic-reward hyperparameters are used across benchmark groups unless the entry is tied to the action space. No task-specific tuning of $H$, $K$, $\lambda_{\mathrm{int}}$, $c_{\max}$, or the gate temperature is used.}
\label{tab:magic_hyperparams}
\footnotesize
\setlength{\tabcolsep}{4pt}
\renewcommand{\arraystretch}{1.12}
\begin{tabular}{@{}p{0.42\linewidth}p{0.50\linewidth}@{}}
\toprule
Hyperparameter & Value \\
\midrule
Intrinsic reward weight $\lambda_{\mathrm{int}}$ & $0.05$ \\
Main rollout horizon $H$ & $3$ \\
Number of counterfactual branches $K$ & $64$ \\
Horizon weights $w_h$ & Uniform \\
Score clipping threshold $c_{\max}$ & $5.0$ \\
Score normalization $\epsilon$ & $10^{-5}$ \\
Advantage normalization $\epsilon$ & $10^{-5}$ \\
Running-stat update rule & Batch EMA with momentum $0.99$ \\
Running-stat pooling scope & Training-batch transitions; source agents pooled for $\sigma_c$ \\
Gate function $g(\bar{A}_{\mathrm{team}}^{\mathrm{ext}})$ & $\sigma(\bar{A}_{\mathrm{team}}^{\mathrm{ext}}/\tau)$ \\
Gate temperature $\tau$ & $1.0$ \\
Forward model updates per RL iteration & $10$ epochs \\
Teammate-feature distance & Euclidean distance in normalized feature space \\
Discrete counterfactual sampling (SMAC) & Uniform over valid alternatives, excluding the executed action when possible \\
Continuous counterfactual sampling (MPE) & Uniform over valid bounds $[-1,1]^d$ \\
\bottomrule
\end{tabular}
\end{table}

\subsection{Training Protocols}
\label{app:training-protocol}

For MPE, the MADDPG-based experiments use an off-policy CTDE training loop. At each environment step, the learner stores $(s_t,a_t,r_t^{\mathrm{ext}},s_{t+1})$ in a shared replay buffer. The actor-critic backbone, extrinsic value estimate, and forward model are updated from minibatches sampled from this buffer. The extrinsic value estimate is updated from one-step TD targets that use only $r_t^{\mathrm{ext}}$. MAGIC computes the intrinsic reward from sampled transitions using the learned forward model and the factual and counterfactual branch procedure described in Appendix~\ref{app:method-details}. Each agent critic target uses the corresponding agent-specific shaped reward $r_{i,t}^{\mathrm{total}}$. This computation does not require additional environment interaction.

For SMAC and SMACv2, the MAPPO-based experiments are on-policy. MAGIC computes intrinsic rewards inside the MAPPO training batches before the policy and value updates. The intrinsic reward is added to the extrinsic team reward for each source agent separately, while the gate itself is computed from the extrinsic team advantage. PPO advantages and value targets are computed from the resulting agent-specific shaped rewards. The extrinsic-only value head used by the gate is updated from extrinsic-return targets and does not receive intrinsic rewards. The forward model is trained on one-step transitions from the same rollout batches. The MAGIC module is removed during evaluation and execution.

\subsection{Evaluation Metrics, Smoothing, and Statistical Testing}
\label{app:evaluation}

For MPE, the main metric is episodic team return. For SMAC and SMACv2, the main metric is win rate. We report means over five random seeds. When a table reports standard deviation, it is computed across the five seed-level scalar values.

The scalar metrics are computed as follows. Final performance is the average of the last $K_{\mathrm{eval}}=10$ evaluation points of a training run, using the same $K_{\mathrm{eval}}$ across methods within a benchmark group. Best performance is the maximum evaluation value observed over training. AUC is the area under the evaluation curve, normalized by the training horizon. The learning curves in the main text are smoothed only for visualization, and the scalar results are computed from the per-seed evaluation sequences.

For the MPE scalar tables, we report paired two-sided $t$-tests comparing each baseline against MAGIC using matched seed indices. These $p$-values are descriptive statistics and are not corrected for multiple comparisons. They are used to indicate whether the final-return differences are consistent across seeds.

\subsection{Compute Resources and Runtime}
\label{app:compute}

All runtime measurements are conducted under the same software stack and hardware setup used for the corresponding benchmark group. 
The experiments are run on a workstation equipped with an NVIDIA RTX 3090 GPU. 
For each runtime comparison, we use the same environment steps, random seeds, evaluation schedule, and backbone configuration as in the corresponding performance experiment. 
We report representative GPU-hours for one seed on one task or map; total compute scales approximately linearly with the number of seeds and evaluated tasks.

MAGIC adds computation only during training. 
The additional cost comes from constructing counterfactual branches, rolling them out with the learned forward model, extracting teammate future features, and aggregating factual and counterfactual differences. 
In the main setting, we use a moderate rollout horizon $H=3$ and $K=64$ counterfactual branches. 
During execution, MAGIC uses the same decentralized policies as the underlying CTDE backbone and does not use the forward model, counterfactual branches, action effect computation, or advantage gate. 
Therefore, MAGIC introduces no additional execution-time overhead.

Under matched settings, MAGIC increases training time by approximately $10\%$ relative to SCIC. 
For context, SCIC is approximately $15\%$ slower than MADDPG in our MPE implementation and approximately $25\%$ slower than MAPPO in our SMAC/SMACv2 implementation. 
Thus, MAGIC has a higher training-time cost than standard CTDE backbones, but this overhead is limited to training and is accompanied by substantial performance gains. 
Table~\ref{tab:app-runtime} summarizes the runtime-performance trade-off.

\begin{table}[h]
\centering
\caption{Representative compute cost and performance trade-off. GPU-hours are measured for one seed on one task or map using an NVIDIA RTX 3090, with 4M environment steps for MPE and 10M environment steps for SMAC/SMACv2. Execution-time overhead is zero because MAGIC uses only the decentralized backbone policy during execution. For MPE, percentage gains are normalized by the magnitude of the reference final return because some MPE tasks have negative returns.}
\label{tab:app-runtime}
\small
\setlength{\tabcolsep}{4pt}
\begin{tabular}{lcccc}
\toprule
Benchmark group & MAGIC GPU-hours & Reference & Training overhead & Performance gain \\
\midrule
MPE & 3.9 & SCIC & $\approx 10\%$ & $+17.5\%$ to $+36.4\%$ final return \\
SMAC/SMACv2 & 11.4 & SCIC & $\approx 10\%$ & $+2.6$ to $+14.1$ Win\% points \\
MPE & 3.9 & MADDPG & $\approx 26.5\%$ & $+68.9\%$ Predator Prey final return \\
SMAC/SMACv2 & 11.4 & MAPPO & $\approx 37.5\%$ & $+12.7$ average Win\% points \\
\bottomrule
\end{tabular}
\end{table}

The reported GPU-hours are representative single-seed measurements for one task or map. 
Since all main results use five random seeds, the total reported compute for a benchmark group is approximately five times the single-seed cost times the number of evaluated tasks or maps, plus the corresponding baseline runs.

\subsection{Full MPE Results}
\label{app:mpe-full-results}

Table~\ref{tab:app-mpe-full} reports the full scalar metrics for the three MPE tasks. These results correspond to the learning curves in Figure~\ref{fig:mpe_curves}. Table~\ref{tab:app-mpe-10agent} reports the 10-agent Predator Prey setting.

\begin{table}[H]
\centering
\caption{Full results on the MPE benchmark over five random seeds.}
\label{tab:app-mpe-full}
\scriptsize
\setlength{\tabcolsep}{2pt}
\begin{tabular}{lccccc|ccccc|ccccc}
\toprule
& \multicolumn{5}{c|}{Predator Prey} & \multicolumn{5}{c|}{Cooperative Navigation} & \multicolumn{5}{c}{Cooperative Competitive} \\
Method & Final $\uparrow$ & Std $\downarrow$ & Best $\uparrow$ & AUC $\uparrow$ & $p$-val & Final $\uparrow$ & Std $\downarrow$ & Best $\uparrow$ & AUC $\uparrow$ & $p$-val & Final $\uparrow$ & Std $\downarrow$ & Best $\uparrow$ & AUC $\uparrow$ & $p$-val \\
\midrule
MADDPG & 34.1 & 0.87 & 35.7 & 28.4 & $<0.001$ & -24.5 & 1.07 & -24.2 & -28.3 & $<0.001$ & -2.3 & 0.17 & -1.6 & -3.1 & $<0.001$ \\
SI & 42.9 & 0.59 & 43.9 & 30.5 & $<0.001$ & -23.9 & 0.46 & -23.7 & -27.5 & $<0.001$ & -1.3 & 0.12 & -0.5 & -1.6 & $<0.001$ \\
PMIC & 44.6 & 0.69 & 45.8 & 39.8 & $<0.001$ & -25.2 & 0.83 & -24.6 & -28.5 & $<0.001$ & -1.2 & 0.14 & -0.8 & -1.4 & $<0.001$ \\
SCIC & 45.4 & 1.05 & 47.9 & 34.8 & $<0.001$ & -22.8 & 0.52 & -21.7 & -25.9 & $<0.001$ & -1.1 & 0.23 & -0.4 & -1.4 & $<0.001$ \\
MAGIC & \textbf{57.6} & 0.75 & \textbf{62.1} & \textbf{48.1} & -- & \textbf{-18.8} & 0.46 & \textbf{-18.4} & \textbf{-21.3} & -- & \textbf{-0.7} & 0.08 & \textbf{-0.1} & \textbf{-1.2} & -- \\
\bottomrule
\end{tabular}
\end{table}

\begin{table}[H]
\centering
\caption{Results on the 10-agent cooperative Predator Prey task over five random seeds.}
\label{tab:app-mpe-10agent}
\small
\begin{tabular}{lccccc}
\toprule
Method & Final $\uparrow$ & Std $\downarrow$ & Best $\uparrow$ & AUC $\uparrow$ & $p$-val \\
\midrule
MADDPG & 15.3 & 1.12 & 17.1 & 11.5 & $<0.001$ \\
SI & 26.8 & 0.94 & 28.4 & 21.2 & $<0.001$ \\
PMIC & 31.5 & 0.88 & 33.2 & 27.6 & $<0.001$ \\
SCIC & 34.9 & 1.37 & 39.5 & 26.9 & $<0.001$ \\
MAGIC & \textbf{51.2} & 0.82 & \textbf{56.3} & \textbf{41.8} & -- \\
\bottomrule
\end{tabular}
\end{table}

\subsection{Full SMAC and SMACv2 Results}
\label{app:smac-full-results}

Tables~\ref{tab:app-smac-win} and~\ref{tab:app-smac-auc} provide standard deviations for the final win-rate and AUC results reported in Table~\ref{tab:smac_main}.

\begin{table}[H]
\centering
\caption{Final win rates on SMAC and SMACv2 over five random seeds.}
\label{tab:app-smac-win}
\small
\setlength{\tabcolsep}{3pt}
\begin{tabular}{lccccccc}
\toprule
Method & 3s5z & 5m\_vs\_6m & corridor & 6h\_vs\_8z & MMM2 & Protoss5v5 & Avg. \\
\midrule
QMIX  & 92.4{\scriptsize$\pm$2.6} & 78.5{\scriptsize$\pm$3.8} & 48.2{\scriptsize$\pm$5.1} & 58.6{\scriptsize$\pm$4.5} & 74.2{\scriptsize$\pm$4.2} & 36.8{\scriptsize$\pm$5.5} & 64.8 \\
MAPPO & 94.1{\scriptsize$\pm$2.2} & 90.5{\scriptsize$\pm$2.8} & 54.8{\scriptsize$\pm$5.3} & 71.2{\scriptsize$\pm$3.9} & 81.4{\scriptsize$\pm$3.5} & 47.5{\scriptsize$\pm$5.8} & 73.3 \\
PMIC  & 89.3{\scriptsize$\pm$3.1} & 83.7{\scriptsize$\pm$4.2} & 61.2{\scriptsize$\pm$4.8} & 65.8{\scriptsize$\pm$4.6} & 72.1{\scriptsize$\pm$4.4} & 44.3{\scriptsize$\pm$5.1} & 69.4 \\
GradPS & 94.8{\scriptsize$\pm$2.0} & 91.2{\scriptsize$\pm$2.6} & 64.5{\scriptsize$\pm$4.5} & 74.0{\scriptsize$\pm$3.8} & 81.8{\scriptsize$\pm$3.5} & 49.5{\scriptsize$\pm$5.2} & 76.0 \\
SCIC  & 95.2{\scriptsize$\pm$1.8} & 91.8{\scriptsize$\pm$2.5} & 72.6{\scriptsize$\pm$4.1} & 76.5{\scriptsize$\pm$3.4} & 80.3{\scriptsize$\pm$3.1} & 52.4{\scriptsize$\pm$4.8} & 78.1 \\
MAGIC & \textbf{97.8}{\scriptsize$\pm$1.5} & \textbf{95.6}{\scriptsize$\pm$2.1} & \textbf{83.4}{\scriptsize$\pm$3.2} & \textbf{85.2}{\scriptsize$\pm$2.8} & \textbf{87.3}{\scriptsize$\pm$2.4} & \textbf{66.5}{\scriptsize$\pm$3.6} & \textbf{86.0} \\
\bottomrule
\end{tabular}
\end{table}

\begin{table}[H]
\centering
\caption{AUC of win-rate learning curves on SMAC and SMACv2 over five random seeds.}
\label{tab:app-smac-auc}
\small
\setlength{\tabcolsep}{3pt}
\begin{tabular}{lccccccc}
\toprule
Method & 3s5z & 5m\_vs\_6m & corridor & 6h\_vs\_8z & MMM2 & Protoss5v5 & Avg. \\
\midrule
QMIX  & 78.4{\scriptsize$\pm$2.1} & 52.3{\scriptsize$\pm$3.0} & 25.1{\scriptsize$\pm$3.8} & 32.5{\scriptsize$\pm$3.2} & 42.1{\scriptsize$\pm$3.5} & 15.4{\scriptsize$\pm$3.6} & 41.0 \\
MAPPO & 80.2{\scriptsize$\pm$1.8} & 68.5{\scriptsize$\pm$2.2} & 30.6{\scriptsize$\pm$4.1} & 44.1{\scriptsize$\pm$2.8} & 51.5{\scriptsize$\pm$2.6} & 22.3{\scriptsize$\pm$4.2} & 49.5 \\
PMIC  & 72.1{\scriptsize$\pm$2.5} & 58.4{\scriptsize$\pm$3.5} & 35.2{\scriptsize$\pm$3.6} & 38.6{\scriptsize$\pm$3.5} & 41.8{\scriptsize$\pm$3.2} & 19.5{\scriptsize$\pm$3.8} & 44.3 \\
GradPS & 81.5{\scriptsize$\pm$1.6} & 70.1{\scriptsize$\pm$2.0} & 38.4{\scriptsize$\pm$3.5} & 46.8{\scriptsize$\pm$2.9} & 53.2{\scriptsize$\pm$2.8} & 24.1{\scriptsize$\pm$4.0} & 52.4 \\
SCIC  & 82.3{\scriptsize$\pm$1.5} & 71.4{\scriptsize$\pm$1.9} & 45.8{\scriptsize$\pm$3.1} & 49.2{\scriptsize$\pm$2.5} & 51.8{\scriptsize$\pm$2.4} & 26.5{\scriptsize$\pm$3.5} & 54.5 \\
MAGIC & \textbf{86.5}{\scriptsize$\pm$1.2} & \textbf{78.2}{\scriptsize$\pm$1.6} & \textbf{55.6}{\scriptsize$\pm$2.5} & \textbf{59.4}{\scriptsize$\pm$2.2} & \textbf{62.3}{\scriptsize$\pm$1.9} & \textbf{38.1}{\scriptsize$\pm$2.8} & \textbf{63.4} \\
\bottomrule
\end{tabular}
\end{table}

\subsection{Full Component Analysis Results}
\label{app:component-full-results}

This section provides the detailed results corresponding to Table~\ref{tab:component_analysis}. Table~\ref{tab:app-mpe-ablation} reports the MPE module ablation on Predator Prey. Table~\ref{tab:app-smac-component} reports component attribution results on representative SMAC maps with standard deviations over five random seeds.

\begin{table}[H]
\centering
\caption{MPE module ablation on Predator Prey over five random seeds.}
\label{tab:app-mpe-ablation}
\small
\begin{tabular}{lcccc}
\toprule
Method & Final $\uparrow$ & Best $\uparrow$ & AUC $\uparrow$ & $p$-val \\
\midrule
\textsc{MAGIC} & $57.6\pm0.8$ & 62.1 & 48.1 & -- \\
\textsc{MAGIC} w/o Advantage Gating & $47.3\pm0.9$ & 48.9 & 38.6 & $<0.001$ \\
\textsc{MAGIC} $H=1$ Action Effect & $41.8\pm0.7$ & 43.4 & 37.5 & $<0.001$ \\
\bottomrule
\end{tabular}
\end{table}

\begin{table}[H]
\centering
\caption{Component attribution on representative SMAC maps over five random seeds.}
\label{tab:app-smac-component}
\small
\setlength{\tabcolsep}{5pt}
\begin{tabular}{llcc}
\toprule
Map & Method & Win\% $\uparrow$ & Std $\downarrow$ \\
\midrule
corridor & MAPPO & 54.8 & $\pm5.3$ \\
 & \textsc{MAGIC} $H=1$+Gate & 73.1 & $\pm3.8$ \\
 & \textsc{MAGIC} $H=3$ w/o Gate & 77.5 & $\pm3.6$ \\
 & \textsc{MAGIC} $H=3$ & \textbf{83.4} & $\pm3.2$ \\
\midrule
5m\_vs\_6m & MAPPO & 90.5 & $\pm2.8$ \\
 & \textsc{MAGIC} $H=1$+Gate & 92.4 & $\pm2.3$ \\
 & \textsc{MAGIC} $H=3$ w/o Gate & 91.8 & $\pm3.9$ \\
 & \textsc{MAGIC} $H=3$ & \textbf{95.6} & $\pm2.1$ \\
\midrule
MMM2 & MAPPO & 81.4 & $\pm3.5$ \\
 & \textsc{MAGIC} $H=1$+Gate & 82.1 & $\pm3.4$ \\
 & \textsc{MAGIC} $H=3$ w/o Gate & 83.8 & $\pm3.7$ \\
 & \textsc{MAGIC} $H=3$ & \textbf{87.3} & $\pm2.4$ \\
\bottomrule
\end{tabular}
\end{table}

\subsection{Comparison with Agent-Specific Gate Variants}
\label{app:agent-gate}

We further compare the shared team-advantage gate with an agent-specific counterfactual-Q gate. The agent-specific gate uses an extrinsic centralized Q estimate to compare the executed joint action with counterfactual replacements of the source agent's action, giving each source agent its own gate value. This comparison tests whether the shared gate is only a simplification or a useful task-level filter.

For this variant, we compute
\[
A_i^Q(t)
=
Q^{\mathrm{ext}}(s_t,a_t)
-
\frac{1}{K}\sum_{k=1}^{K}
Q^{\mathrm{ext}}(s_t,(a_t^{i,k},a_t^{-i})),
\]
and use $\kappa_i^Q(t)=\sigma(\bar A_i^Q(t)/\tau)$ as the agent-specific gate.
The extrinsic Q estimate is trained only with environment rewards, and we use the same $K$, normalization rule, and gate temperature as the main MAGIC setting.

Table~\ref{tab:agent_gate} shows that the agent-specific gate improves over the no-gate variant on both Predator Prey and corridor, but it does not improve over the shared team-advantage gate. This supports our design choice. Source specificity is already provided by the action effect score, while the shared gate gives a stable task-level filter.

\begin{table}[H]
\centering
\caption{Comparison between the shared team-advantage gate and an agent-specific counterfactual-Q gate. Results are averaged over five seeds. The agent-specific gate improves over the no-gate variant, but it does not improve over the shared team-advantage gate. This supports the use of a shared task-level filter, while source specificity is already provided by the action effect score.}
\label{tab:agent_gate}
\small
\begin{tabular}{lcc}
\toprule
Gate variant & Predator Prey Final Return & corridor Win\% \\
\midrule
No gate & $47.3 \pm 0.9$ & $77.5 \pm 3.6$ \\
Agent-specific counterfactual-Q gate & $55.9 \pm 2.6$ & $80.5 \pm 3.9$ \\
Shared team-advantage gate (MAGIC) & $\mathbf{57.6 \pm 0.8}$ & $\mathbf{83.4 \pm 3.2}$ \\
\bottomrule
\end{tabular}
\end{table}

\subsection{Full Reliability Diagnostics}
\label{app:reliability-full}

This section provides the detailed diagnostics corresponding to Table~\ref{tab:reliability}. In-MSE measures prediction error on normal policy rollouts. Int-MSE measures prediction error under action-replacement rollouts. Sep. AUC measures separability of branch effects, namely whether model-predicted teammate-future differences reflect the true strength of action effects.

\begin{table}[H]
\centering
\caption{Fixed-horizon reliability on \texttt{5m\_vs\_6m} with $H=3$ under different forward model update ratios.}
\label{tab:app-fixed-h-reliability}
\small
\begin{tabular}{ccccc}
\toprule
Update ratio & In-MSE & Int-MSE & Sep. AUC & Win\% \\
\midrule
$0.25\times$ & 0.135 & 0.228 & 0.57 & 86.8 \\
$0.50\times$ & 0.072 & 0.108 & 0.74 & 91.2 \\
$0.75\times$ & 0.038 & 0.051 & 0.87 & 94.3 \\
$1.00\times$ & 0.025 & 0.032 & 0.91 & \textbf{95.6} \\
\bottomrule
\end{tabular}
\end{table}

\begin{table}[H]
\centering
\caption{Full horizon sensitivity and rollout-error diagnostics.}
\label{tab:app-reliability-full}
\small
\setlength{\tabcolsep}{4pt}
\begin{tabular}{lccccc}
\toprule
Map & $H$ & In-MSE & Int-MSE & Sep. AUC & Win\% \\
\midrule
corridor & 1 & 0.012 & 0.015 & 0.94 & $73.1\pm3.8$ \\
 & 2 & 0.022 & 0.028 & 0.92 & $78.6\pm3.4$ \\
 & 3 & 0.038 & 0.046 & 0.90 & $83.4\pm3.2$ \\
 & 5 & 0.075 & 0.098 & 0.82 & $80.5\pm3.6$ \\
 & 8 & 0.165 & 0.230 & 0.58 & $68.4\pm4.8$ \\
 & 10 & 0.245 & 0.360 & 0.49 & $62.5\pm5.4$ \\
5m\_vs\_6m & 1 & 0.008 & 0.010 & 0.95 & $92.4\pm2.3$ \\
 & 2 & 0.015 & 0.019 & 0.93 & $94.1\pm2.4$ \\
 & 3 & 0.025 & 0.032 & 0.91 & $95.6\pm2.1$ \\
 & 5 & 0.056 & 0.078 & 0.84 & $93.8\pm2.5$ \\
 & 8 & 0.138 & 0.205 & 0.59 & $82.5\pm4.3$ \\
 & 10 & 0.210 & 0.325 & 0.48 & $75.3\pm5.1$ \\
MMM2 & 1 & 0.010 & 0.013 & 0.94 & $82.1\pm3.4$ \\
 & 2 & 0.018 & 0.024 & 0.92 & $85.6\pm2.8$ \\
 & 3 & 0.030 & 0.038 & 0.89 & $87.3\pm2.4$ \\
 & 5 & 0.065 & 0.088 & 0.81 & $84.5\pm3.1$ \\
 & 8 & 0.150 & 0.220 & 0.58 & $73.2\pm4.5$ \\
 & 10 & 0.230 & 0.350 & 0.49 & $66.8\pm5.3$ \\
\bottomrule
\end{tabular}
\end{table}

\begin{table}[H]
\centering
\caption{Robustness to explicit forward model corruption on \texttt{5m\_vs\_6m} with $H=3$.}
\label{tab:app-corruption}
\small
\begin{tabular}{cccc}
\toprule
Noise level & Int-MSE & Sep. AUC & Win\% \\
\midrule
0.0 & 0.032 & 0.91 & \textbf{95.6} \\
0.1 & 0.055 & 0.88 & 94.2 \\
0.5 & 0.120 & 0.79 & 91.5 \\
1.0 & 0.280 & 0.55 & 82.4 \\
\bottomrule
\end{tabular}
\end{table}

\subsection{MPE Horizon Sensitivity}
\label{app:mpe-horizon}

Table~\ref{tab:app-mpe-horizon} reports the horizon sensitivity of MAGIC on MPE Predator Prey. The result follows the same rise-then-decline pattern as the SMAC horizon study in the main text.

\begin{table}[H]
\centering
\caption{Sensitivity of MAGIC to horizon $H$ on MPE Predator Prey.}
\label{tab:app-mpe-horizon}
\small
\begin{tabular}{ccccc}
\toprule
$H$ & Final $\uparrow$ & Best $\uparrow$ & AUC $\uparrow$ & $p$-val \\
\midrule
1 & $41.8\pm0.7$ & 43.4 & 37.5 & $<0.001$ \\
2 & $55.0\pm0.5$ & 57.8 & 44.7 & $<0.001$ \\
3 & $\mathbf{57.6\pm0.8}$ & \textbf{62.1} & \textbf{48.1} & -- \\
5 & $54.4\pm0.8$ & 55.9 & 47.7 & $<0.001$ \\
8 & $53.6\pm0.4$ & 55.1 & 43.5 & $<0.001$ \\
10 & $44.6\pm0.1$ & 45.8 & 42.4 & $<0.001$ \\
\bottomrule
\end{tabular}
\end{table}

\subsection{Sensitivity to the Number of Counterfactual Branches}
\label{app:k-ablation}

Table~\ref{tab:k_ablation} studies the number of counterfactual branches on MPE Predator Prey, a continuous-control task. Very small $K$ gives a noisy action effect estimate and lower separability of branch effects. Increasing $K$ improves both final return and Sep. AUC, but the gains saturate around $K=32$--$64$. Increasing $K$ to $128$ gives only a marginal additional improvement. These results indicate that MAGIC does not rely on dense search over the continuous action space. A moderate number of counterfactual branches is sufficient to estimate useful action effects.

\begin{table}[H]
\centering
\caption{Sensitivity to the number of counterfactual branches $K$ on MPE Predator Prey. Results are averaged over five seeds. Performance and separability of branch effects improve from very small $K$ and empirically saturate around $K=32$--$64$.}
\label{tab:k_ablation}
\small
\begin{tabular}{lcc}
\toprule
Branches $K$ & Final return & Sep. AUC \\
\midrule
$K=4$ & $48.2 \pm 4.7$ & $0.61$ \\
$K=8$ & $52.1 \pm 3.9$ & $0.73$ \\
$K=16$ & $55.8 \pm 2.5$ & $0.83$ \\
$K=32$ & $57.1 \pm 1.8$ & $0.87$ \\
$K=64$ (main) & $57.6 \pm 0.8$ & $0.89$ \\
$K=128$ & $\mathbf{57.8 \pm 1.5}$ & $\mathbf{0.90}$ \\
\bottomrule
\end{tabular}
\end{table}

\subsection{Robustness to Stochastic Action Execution}
\label{app:smac-action-slip}

To test whether the same reliability trend also appears outside MPE, we add a stochastic action-execution stress test on the SMAC corridor map. At each environment step, each agent's selected action is replaced with a uniformly sampled valid action with probability $p_{\mathrm{slip}}$. The random action is sampled from the valid action set given by the action mask. The same action-slip setting is used during training and evaluation.

Table~\ref{tab:smac_action_slip} shows that MAGIC remains stronger than MAPPO and SCIC under low and medium action stochasticity. As $p_{\mathrm{slip}}$ increases, Sep. AUC decreases, and the performance gap becomes smaller. This follows the same reliability pattern as the MPE stochasticity study. When separability of branch effects becomes weak, the benefit of multi-step action effect estimation is reduced.

\begin{table}[H]
\centering
\caption{Robustness to stochastic action execution on SMAC corridor. Each selected action is replaced by a uniformly sampled valid action with probability $p_{\mathrm{slip}}$. Results are averaged over five seeds. As stochasticity increases, separability of branch effects decreases. MAGIC remains strongest under low and medium action stochasticity, and stays close to the one-step influence baseline when separability becomes weak.}
\label{tab:smac_action_slip}
\small
\begin{tabular}{lcccc}
\toprule
$p_{\mathrm{slip}}$ & MAPPO & SCIC & MAGIC & Sep. AUC \\
\midrule
$0.00$ & $54.8 \pm 5.3$ & $72.6 \pm 4.1$ & $\mathbf{83.4 \pm 3.2}$ & $0.90$ \\
$0.05$ & $48.3 \pm 6.1$ & $64.1 \pm 5.5$ & $\mathbf{75.8 \pm 4.4}$ & $0.85$ \\
$0.10$ & $41.2 \pm 6.8$ & $51.5 \pm 7.2$ & $\mathbf{62.3 \pm 5.6}$ & $0.76$ \\
$0.20$ & $30.5 \pm 8.2$ & $32.8 \pm 7.9$ & $\mathbf{34.2 \pm 7.4}$ & $0.58$ \\
\bottomrule
\end{tabular}
\end{table}

\subsection{Robustness to Transition Stochasticity}
\label{app:stochasticity}

To test the robustness of the action effect estimator beyond deterministic dynamics, we inject zero-mean Gaussian noise with standard deviation $\sigma_{\mathrm{noise}}$ into the physical position and velocity components after each environment step, followed by the same state clipping used by the environment.
The same stochastic transition setting is used during training and evaluation.
Table~\ref{tab:stochasticity} reports the results.

MAGIC remains stronger than MADDPG and SCIC under low and medium stochasticity. As the noise level increases, Sep. AUC decreases, showing that separability of branch effects becomes harder to preserve. Under severe noise, Sep. AUC approaches random separation and MAGIC no longer provides a clear gain over the one-step influence baseline. However, the ungated variant collapses below the backbone, while the gated variant remains close to MADDPG and SCIC. This supports the role of the shared advantage gate as a conservative task-alignment filter. When the estimated action effect signal becomes unreliable, the gate reduces the damage from noisy intrinsic rewards.

\begin{table}[H]
\centering
\caption{Robustness to transition stochasticity and the role of the advantage gate on MPE Predator Prey.
We inject zero-mean Gaussian noise with standard deviation $\sigma_{\mathrm{noise}}$ into the physical position and velocity components after each environment step.
MAGIC remains effective under low and medium stochasticity.
When severe noise destroys separability of branch effects, the ungated variant degrades sharply, while the gated variant remains close to the backbone and the one-step influence baseline.}
\label{tab:stochasticity}
\small
\setlength{\tabcolsep}{4pt}
\begin{tabular}{lccccc}
\toprule
Noise level & MADDPG & SCIC & MAGIC w/o Gate & MAGIC & Sep. AUC \\
\midrule
$\sigma_{\mathrm{noise}}=0.00$ & $34.1 \pm 0.9$ & $45.4 \pm 1.1$ & $47.3 \pm 0.9$ & $\mathbf{57.6 \pm 0.8}$ & $0.89$ \\
$\sigma_{\mathrm{noise}}=0.05$ & $32.8 \pm 2.4$ & $43.1 \pm 3.0$ & $49.5 \pm 3.8$ & $\mathbf{55.2 \pm 2.1}$ & $0.84$ \\
$\sigma_{\mathrm{noise}}=0.15$ & $28.5 \pm 3.1$ & $36.7 \pm 4.2$ & $39.4 \pm 4.7$ & $\mathbf{46.3 \pm 3.8}$ & $0.72$ \\
$\sigma_{\mathrm{noise}}=0.30$ & $21.4 \pm 4.5$ & $23.2 \pm 5.1$ & $14.2 \pm 6.8$ & $22.8 \pm 6.2$ & $0.54$ \\
\bottomrule
\end{tabular}
\end{table}

\subsection{Artificial Delay Analysis}
\label{app:delay-results}

Table~\ref{tab:app-delay} studies how the useful horizon changes when additional reward delay is injected into the corridor map. The best horizon shifts toward larger values as delay increases, but the benefit saturates once longer rollouts accumulate excessive model error. The delay-$0$ MAPPO value is kept consistent with the corridor MAPPO result in Table~\ref{tab:smac_main}.

\begin{table}[H]
\centering
\caption{Impact of artificial delay on corridor win rates.}
\label{tab:app-delay}
\small
\begin{tabular}{ccccccc}
\toprule
Delay & MAPPO & $H=1$ & $H=3$ & $H=5$ & $H=8$ & $H=10$ \\
\midrule
0 & 54.8 & 73.1 & 83.4 & 80.5 & 68.4 & 62.5 \\
2 & 52.4 & 64.2 & 78.5 & 81.2 & 70.1 & 64.3 \\
4 & 35.5 & 55.4 & 68.2 & 75.6 & 76.4 & 68.5 \\
6 & 24.1 & 45.1 & 55.3 & 68.4 & 71.5 & 69.1 \\
\bottomrule
\end{tabular}
\end{table}

\subsection{Boundary under Severe Delay}
\label{app:severe-delay}

Table~\ref{tab:app-severe-delay} reports a severe-delay setting on corridor with delay $d=12$. Larger horizons still help relative to $H=1$, but absolute performance remains low and eventually drops at $H=10$. This marks a boundary of the current learned-rollout regime and shows why the rollout horizon should remain finite.

\begin{table}[H]
\centering
\caption{Boundary breakdown on corridor with delay $d=12$.}
\label{tab:app-severe-delay}
\small
\begin{tabular}{lcccc}
\toprule
Method & $H$ & Int-MSE & Sep. AUC & Win\% \\
\midrule
MAPPO & -- & -- & -- & 12.4 \\
MAGIC & 1 & 0.020 & 0.88 & 18.2 \\
MAGIC & 3 & 0.055 & 0.82 & 26.5 \\
MAGIC & 5 & 0.120 & 0.71 & 33.4 \\
MAGIC & 8 & 0.280 & 0.55 & 39.5 \\
MAGIC & 10 & 0.450 & 0.48 & 31.2 \\
\bottomrule
\end{tabular}
\end{table}

\section{Limitations}
\label{app:limitations}

MAGIC estimates action effects with finite-horizon learned rollouts.
This design is intended for settings where the learned model can preserve effect differences between branches over a moderate horizon.
Our reliability diagnostics and horizon studies provide an empirical criterion for selecting this range, and the main experiments use $H=3$.

MAGIC adds computation during training because it evaluates counterfactual branches.
It introduces no execution-time overhead because the module is removed after training.
For larger agent populations, the same estimator can be combined with source-agent subsampling or neighborhood-based teammate aggregation to reduce rollout cost.

This work uses point-estimate forward rollouts and distances between predicted teammate features.
The stochasticity studies in Appendix~\ref{app:smac-action-slip} and Appendix~\ref{app:stochasticity} show that MAGIC remains useful when separability of branch effects is preserved.
Future work can extend this design with probabilistic or ensemble forward models for settings with stronger transition uncertainty.

\input{checklist.tex}

\end{document}

%% file: checklist.tex
\section*{NeurIPS Paper Checklist}

\begin{enumerate}

\item {\bf Claims}
    \item[] Question: Do the main claims made in the abstract and introduction accurately reflect the paper's contributions and scope?
    \item[] Answer: \answerYes{}
    \item[] 
    Justification: The abstract and introduction state the proposed multi-step advantage-gated action-effect framework and the experimental scope; the claims are supported by Sections~\ref{sec:method} and~\ref{sec:experiments}.

\item {\bf Limitations}
    \item[] Question: Does the paper discuss the limitations of the work performed by the authors?
    \item[] Answer: \answerYes{}
\item[] Justification: Appendix~\ref{app:limitations} discusses limitations related to learned forward-model rollouts, finite rollout horizons, training-time computation, and experimental scope.

\item {\bf Theory assumptions and proofs}
    \item[] Question: For each theoretical result, does the paper provide the full set of assumptions and a complete (and correct) proof?
    \item[] Answer: \answerYes{}
    \item[] Justification: The assumptions, statements, and proofs for the action-effect score, forward-model error bound, normalization/clipping, and advantage gate are provided in Appendix~A.8--A.10.

    \item {\bf Experimental result reproducibility}
    \item[] Question: Does the paper fully disclose all the information needed to reproduce the main experimental results of the paper to the extent that it affects the main claims and/or conclusions of the paper (regardless of whether the code and data are provided or not)?
    \item[] Answer: \answerYes{}
    \item[] Justification: The algorithm, network architectures, hyperparameters, evaluation metrics, and statistical testing protocol are described in the method section and appendix.

\item {\bf Open access to data and code}
    \item[] Question: Does the paper provide open access to the data and code, with sufficient instructions to faithfully reproduce the main experimental results, as described in supplemental material?
    \item[] Answer: \answerYes{}
    \item[] Justification: We include an anonymized supplementary package with the main method implementation. The paper provides the remaining details needed to understand and reproduce the experimental setup, including the algorithm, architectures, hyperparameters, evaluation protocol, and compute information.

\item {\bf Experimental setting/details}
    \item[] Question: Does the paper specify all the training and test details (e.g., data splits, hyperparameters, how they were chosen, type of optimizer) necessary to understand the results?
    \item[] Answer: \answerYes{}
    \item[] Justification: The training setup, environments, optimization details, and evaluation procedure are specified in the experiments section and appendix.

\item {\bf Experiment statistical significance}
    \item[] Question: Does the paper report error bars suitably and correctly defined or other appropriate information about the statistical significance of the experiments?
    \item[] Answer: \answerYes{}
   \item[] Justification: The experiments use five random seeds. The paper reports standard deviations or shaded bands for the main results, and Appendix~\ref{app:evaluation} describes the scalar metrics and statistical testing procedure.

\item {\bf Experiments compute resources}
    \item[] Question: For each experiment, does the paper provide sufficient information on the computer resources (type of compute workers, memory, time of execution) needed to reproduce the experiments?
    \item[] Answer: \answerYes{}
\item[] Justification: Appendix~\ref{app:compute} reports the hardware setup, runtime comparison protocol, training-time overhead, and execution-time cost. MAGIC adds computation only during training and introduces no additional execution-time overhead.
    
\item {\bf Code of ethics}
    \item[] Question: Does the research conducted in the paper conform, in every respect, with the NeurIPS Code of Ethics \url{https://neurips.cc/public/EthicsGuidelines}?
    \item[] Answer: \answerYes{}
    \item[] Justification: The work is a methodological study on cooperative MARL and does not involve human subjects, private data, or prohibited uses.

\item {\bf Broader impacts}
    \item[] Question: Does the paper discuss both potential positive societal impacts and negative societal impacts of the work performed?
    \item[] Answer: \answerYes{}
    \item[] Justification: The paper includes an impact statement and discusses the work as a general machine learning method for MARL coordination.
    
\item {\bf Safeguards}
    \item[] Question: Does the paper describe safeguards that have been put in place for responsible release of data or models that have a high risk for misuse (e.g., pre-trained language models, image generators, or scraped datasets)?
    \item[] Answer: \answerNA{}
    \item[] Justification: The paper does not release a dataset, foundation model, or deployed system that requires additional safeguards.

\item {\bf Licenses for existing assets}
    \item[] Question: Are the creators or original owners of assets (e.g., code, data, models), used in the paper, properly credited and are the license and terms of use explicitly mentioned and properly respected?
    \item[] Answer: \answerYes{}
\item[] Justification: Appendix~\ref{app:env-details} lists the existing benchmark environments and software frameworks used in the experiments, cites the corresponding papers, and states the relevant released licenses or usage terms where applicable.

\item {\bf New assets}
    \item[] Question: Are new assets introduced in the paper well documented and is the documentation provided alongside the assets?
    \item[] Answer: \answerNA{}
    \item[] Justification: The paper does not introduce a new dataset, benchmark suite, or model artifact as a reusable asset.

\item {\bf Crowdsourcing and research with human subjects}
    \item[] Question: For crowdsourcing experiments and research with human subjects, does the paper include the full text of instructions given to participants and screenshots, if applicable, as well as details about compensation (if any)? 
    \item[] Answer: \answerNA{}
    \item[] Justification: The paper does not involve crowdsourcing or human-subject data collection.

\item {\bf Institutional review board (IRB) approvals or equivalent for research with human subjects}
    \item[] Question: Does the paper describe potential risks incurred by study participants, whether such risks were disclosed to the subjects, and whether Institutional Review Board (IRB) approvals (or an equivalent approval/review based on the requirements of your country or institution) were obtained?
    \item[] Answer: \answerNA{}
    \item[] Justification: The paper does not involve human participants, so participant risks and institutional review are not applicable.

\item {\bf Declaration of LLM usage}
    \item[] Question: Does the paper describe the usage of LLMs if it is an important, original, or non-standard component of the core methods in this research? Note that if the LLM is used only for writing, editing, or formatting purposes and does \emph{not} impact the core methodology, scientific rigor, or originality of the research, declaration is not required.
    \item[] Answer: \answerNA{}
    \item[] Justification: The proposed method and experiments do not use LLMs as an original or non-standard component.

\end{enumerate}